\documentclass{article}

\usepackage{amsthm}
\usepackage{amsmath}
\usepackage{amssymb}
\usepackage{amsfonts}
\usepackage{bussproofs}
\usepackage{proof}
\usepackage[usenames,dvipsnames]{color}
\usepackage{stmaryrd}
\usepackage{latexsym}
\usepackage{url}

\theoremstyle{theorem}

\newtheorem{proposition}{Proposition}[section]
\newtheorem{theorem}{Theorem}[section]
\theoremstyle{definition}
\newtheorem{df}{Definition}[section]
\newtheorem{definition}[df]{Definition}
\newtheorem{example}{Example}[section]

\newcommand{\CL}{\ensuremath{\mathrm{CL}}}
\newcommand{\PR}{\ensuremath{\mathrm{PR}}}
\newcommand{\cl}{\ensuremath{\mathrm{cl}}}
\newcommand{\pr}{\ensuremath{\mathrm{pr}}}
\newcommand{\LK}{\ensuremath{\mathbf{LK}}}
\newcommand{\LKS}{\ensuremath{\mathbf{LKS}_\Ecal}}
\newcommand{\LKSI}{\LKIE}
\newcommand{\seq}{\vdash}

\newcommand{\bigor}{\bigvee}
\newcommand{\la}{\leftarrow}
\newcommand{\ra}{\rightarrow}
\newcommand{\impl}{\supset}

\newcommand{\goesto}{\leftarrow}

\newcommand{\CERESs}{{\rm CERES}_s}
\newcommand{\Acal}{{\cal A}}
\newcommand{\Ecal}{{\cal E}}
\newcommand{\Rcal}{{\cal R}}
\newcommand{\Scal}{{\cal S}}
\newcommand{\LKE}{\ensuremath{\mathbf{LK_\Ecal}}}
\newcommand{\LKIE}{\ensuremath{\mathbf{LKI_\Ecal}}}
\newcommand{\union}{\cup}
\newcommand{\fhat}{\hat{f}}
\newcommand{\ES}{{\it ES}}
\newcommand{\res}{{\it res}}
\newcommand{\N}{\mathbb{N}}

\newcommand{\Ccal}{{\cal C}}

\newcommand{\IN}{\subseteq}

\newcommand{\Null}{\bar{0}}
\newcommand{\ioi}{\leftrightarrow}
\newcommand{\ass}{\leftarrow}
\newcommand{\Vso}{V_2}
\newcommand{\Vcs}{V_{{\rm clset}}}
\newcommand{\CST}{{\rm CST}}
\newcommand{\Plus}{\oplus}
\newcommand{\Times}{\otimes}
\newcommand{\vcst}{v_{{\rm cst}}}
\newcommand{\vc}{v_{{\rm c}}}
\newcommand{\dom}{{\it dom}}
\newcommand{\clauses}{{\it CLAUSES}}
\newcommand{\CS}{{\it CS}}

\newcommand{\vphi}{\varphi}

\newcommand{\all}{\forall}
\newcommand{\ex}{\exists}
\newcommand{\Or}{\vee}
\newcommand{\Pb}{\bar{P}}
\newcommand{\Qb}{\bar{Q}}
\newcommand{\nats}{\mathbb{N}}

\newcommand{\rw}{\twoheadrightarrow} 
\newcommand{\ind}{\mathrm{IND}} 
\newcommand{\vars}{\mathrm{V}} 
\def\RCS$#1: #2 ${\expandafter\def\csname RCS#1\endcsname{#2}}
\RCS$Revision: 1.45 $ 
\RCS$Date: 2012-11-22 13:23:21 $

\title{CERES for First-Order Schemata\thanks{Supported by the project I383 of the Austrian Science Fund.}}
\author{Cvetan Dunchev \and Alexander Leitsch \and Mikheil Rukhaia \and Daniel Weller}

\begin{document}

\maketitle


\begin{abstract}
The cut-elimination method CERES (for first- and higher-order classical logic) is based on the notion of a characteristic clause set, which is extracted from an {\LK}-proof and is always unsatisfiable. A resolution refutation of this clause set can be used as a skeleton for a proof with atomic cuts only (atomic cut normal form). This is achieved by replacing clauses from the resolution refutation by the corresponding projections of the original proof.

We present a generalization of CERES (called $\CERESs$) to first-order proof schemata and define a schematic version of the sequent calculus called \LKS, and a notion of proof schema based on primitive recursive definitions. A method is developed to extract schematic characteristic clause sets and schematic projections from these proof schemata. We also define a schematic resolution calculus for refutation of schemata of clause sets, which can be applied to refute the schematic characteristic clause sets. Finally the projection schemata and resolution schemata are plugged together and a schematic representation of the atomic cut normal forms is obtained. A major benefit of $\CERESs$ is the extension of cut-elimination to inductively defined proofs: we compare $\CERESs$ with standard calculi using induction rules and demonstrate that $\CERESs$ is capable of performing cut-elimination where traditional methods fail. The algorithmic handling of $\CERESs$ is supported by a recent extension of the CERES system.

\smallskip
\noindent \textbf{Keywords:} Cut-elimination, induction, schemata, resolution.
\end{abstract}

\section{Introduction}

Cut-elimination was originally introduced by G.~Gentzen in~\cite{Gentzen1935}
as a theoretical tool from which results like decidability and consistency could
be proven. Cut-free proofs are computationally explicit objects from which interesting
information such as Herbrand disjunctions and interpolants
can be easily extracted. When viewing formal proofs as a model for mathematical proofs,
cut-elimination corresponds to the removal of lemmas, which leads to interesting applications
(such as one described below).

For such applications to mathematical proofs, the cut-elimination method CERES (cut-elimination
by resolution) was developed. It essentially reduces cut-elimination for a proof $\pi$
to a theorem proving problem: the refutation of the {\em characteristic clause set $\CL(\pi)$}.
Given a resolution refutation $\gamma$ of $\CL(\pi)$, an essentially cut-free proof
can be constructed by a proof-theoretic transformation.

It is well-known that cut-elimination in standard calculi of arithmetic, containing an induction rule, is impossible in general~\cite{Takeuti1975} (see also~\cite{McDowell2000,Baelde2007,Tait1968} for other approaches to inductive reasoning using induction
rules). In fact, if $\varphi$ is a proof of a sequent $S\colon \Gamma \seq \forall x.A(x)$, where an induction rule occurs over a cut, the cut cannot be shifted over the induction rule and thus cannot be eliminated. This is not a feature of a specific cut-elimination method, but ,even in principle, inductive proofs require lemmata which cannot be eliminated. When we consider, on the other hand, the infinite sequence of proofs $\varphi_n$ of $S_n\colon \Gamma \seq A(n)$, every of these proofs enjoys cut-elimination. One could hope that, with a sufficiently nice finite description of the infinite sequence $\varphi_n$, a finite description of a sequence of corresponding cut-free proofs comes within reach. The subject of this paper is to find appropriate finite representations of such proof sequences and to develop a formalism to represent sequences of corresponding cut-free proofs. It turned out that, to this aim, the method CERES (cut-elimination by resolution) is more suitable than the traditional reductive method. Via the above transformation we obtain a method of cut-elimination for inductive proofs, which is impossible in ordinary arithmetic calculi.
The approach to describing infinite sequences of proofs 
we take will be based on {\em proof links} which serve as formal placeholders for proofs. Related
approaches are found in the literature on {\em cyclic proofs}~\cite{Sprenger2003,Brotherston2005}.

The present work was also motivated by an application of CERES to (a formalization of) a mathematical proof:
F\"urstenberg's proof of the infinity of primes~\cite{Aigner1999,Baaz2008}. The proof was formalized as a sequence of proofs $\varphi_n$ 
showing that the assumption that there exist exactly $n$ primes is contradictory. The application was
performed in a semi-automated way: $\CL(\varphi_n)$ was computed for some small values of $n$ and from
this, a general schema $\CL(\varphi_n)$ was constructed and subsequently analyzed by hand. The analysis
finally showed that from F\"urstenberg's proof, which makes use of topological concepts,
Euclid's elementary proof could be obtained by cut-elimination.  

The analysis of the F\"urstenberg proof described above reveals the need for the development of a formal language for handling schemata. First of all one can compute the schema $\CL(\varphi_n)$ directly from the specification of the proofs $\varphi_n$, thus paving the way for formal verification of schematic cut-elimination. Formal schematic resolution calculi could provide a tool for interactively developing and verifying refutations of clause schemata. Furthermore, on the basis of these definitions, software tools for cut-elimination in the presence of induction can be developed.

This paper is structured in the following way: In Section~\ref{sec:ceres} we give a short description of the method CERES  for first-order logic. In Section~\ref{sec:motivation} we define an inductive proof $\varphi$ of a sequent $S\colon \Gamma \seq \forall x.A(x)$ not admitting cut-elimination and show informally, how we can obtain a uniform description of the proofs $\varphi_n$ and of the corresponding sequence of cut-free proofs. The rest of the paper is devoted to the development of a formal machinery realizing the methodology described above. In Sections~\ref{sec:schemlang} and~\ref{sec:schemproof} we introduce schematic first-order languages for formulas and proofs. In Section~\ref{sec:resschema} we introduce the concept of clause-schemata and clause-set schemata and develop a schematic resolution calculus. In Section~\ref{sec:schemceres} we extend the first-order CERES method to proof schemata.    

\section{Notations and Definitions}

Our notions of proof will all be based on the usual classical sequent calculus {\LK}: An 
expression of the form $\Gamma \seq \Delta$, where $\Gamma$ and $\Delta$ are
multisets of formulas, is called a {\em sequent}.

We define some simple operations on sequents: let $S\colon \Gamma \seq \Delta$ and $S'\colon \Pi \seq \Lambda$ be sequents; we define $S \circ S'$ (the merge of $S$
and $S'$) as $\Gamma,\Pi \seq \Delta,\Lambda$. Let $\Scal$ and $\Scal'$ be sets of sequents then
$$\Scal \times \Scal' = \{S \circ S' \mid S \in \Scal, S' \in \Scal'\}.$$
%
%
%
The rules of the sequent calculus {\LK} for first-order logic are the following:
\begin{enumerate}
\item Logical rules:
\begin{itemize}
\item $\neg$ introduction \\
\begin{center}
\AxiomC{$\Gamma \vdash \Delta, A$}
\RightLabel{$\neg \colon l$} 
\UnaryInfC{$\neg A, \Gamma \vdash \Delta$} 
\DisplayProof 
$\qquad$ and $\qquad$ 
\AxiomC{$A, \Gamma \vdash \Delta$} 
\RightLabel{$\neg \colon r$} 
\UnaryInfC{$\Gamma \vdash \Delta, \neg A$} 
\DisplayProof \newline
\end{center}

\item $\land$ introduction \\
\begin{center}
\AxiomC{$A, \Gamma \vdash \Delta$} 
\RightLabel{$\land \colon l1$} 
\UnaryInfC{$A \land B, \Gamma \vdash \Delta$} 
\DisplayProof 
$\qquad$ and $\qquad$ 
\AxiomC{$B, \Gamma \vdash \Delta$} 
\RightLabel{$\land \colon l2$} 
\UnaryInfC{$A \land B, \Gamma \vdash \Delta$} 
\DisplayProof \newline
\end{center}
\begin{center}
\AxiomC{$\Gamma \vdash \Delta, A$} 
\AxiomC{$\Pi \vdash \Lambda, B$} 
\RightLabel{$\land \colon r$} 
\BinaryInfC{$\Gamma, \Pi \vdash \Delta, \Lambda, A \land B$}
\DisplayProof \newline
\end{center}

\item $\lor$ introduction \\
\begin{center}
\AxiomC{$A, \Gamma \vdash \Delta$} 
\AxiomC{$B, \Pi \vdash \Lambda$} 
\RightLabel{$\lor \colon l$} 
\BinaryInfC{$A \lor B, \Gamma, \Pi \vdash \Delta, \Lambda$}
\DisplayProof \newline
\end{center}
\begin{center}
\AxiomC{$\Gamma \vdash \Delta, A$} 
\RightLabel{$\lor \colon r1$} 
\UnaryInfC{$\Gamma \vdash \Delta, A \lor B$} 
\DisplayProof 
$\qquad$ and $\qquad$ 
\AxiomC{$\Gamma \vdash \Delta, B$} 
\RightLabel{$\lor \colon r2$} 
\UnaryInfC{$\Gamma \vdash \Delta, A \lor B$} 
\DisplayProof \newline
\end{center}

\item $\impl$ introduction \\
\begin{center}
\AxiomC{$\Gamma \vdash \Delta, A$} 
\AxiomC{$B, \Pi \vdash \Lambda$} 
\RightLabel{$\impl \colon l$} 
\BinaryInfC{$A \impl B, \Gamma, \Pi \vdash \Delta, \Lambda$}
\DisplayProof 
$\quad$ and $\quad$ 
\AxiomC{$A, \Gamma \vdash \Delta, B$} 
\RightLabel{$\impl \colon r$} 
\UnaryInfC{$\Gamma \vdash \Delta, A \impl B$} 
\DisplayProof \newline
\end{center}

\item $\forall$ introduction \\
\begin{center}
\AxiomC{$\Gamma \vdash \Delta, A(\alpha)$} 
\RightLabel{$\forall \colon r$} 
\UnaryInfC{$\Gamma \vdash \Delta, \forall x\;A(x)$} 
\DisplayProof 
$\qquad$ and $\qquad$
\AxiomC{$A(t), \Gamma \vdash \Delta$} 
\RightLabel{$\forall \colon l$} 
\UnaryInfC{$\forall x\;A(x),\Gamma \vdash \Delta$} 
\DisplayProof \newline
\end{center}
where $\alpha$ is a variable of appropriate sort not occuring in $\Gamma,\Delta,A(x)$.

\item $\exists$ introduction \\
\begin{center}
\AxiomC{$\Gamma \vdash \Delta, A(t)$} 
\RightLabel{$\exists \colon r$} 
\UnaryInfC{$\Gamma \vdash \Delta, \exists x\;A(x)$} 
\DisplayProof 
$\qquad$ and $\qquad$
\AxiomC{$A(\alpha), \Gamma \vdash \Delta$} 
\RightLabel{$\exists \colon l$} 
\UnaryInfC{$\exists x\;A(x),\Gamma \vdash \Delta$} 
\DisplayProof \newline
\end{center}
where $\alpha$ is a variable of appropriate sort not occuring in $\Gamma,\Delta,A(x)$.
\end{itemize}

\item Structural rules:
\begin{itemize}
\item Weakening rules: \\
\begin{center}
\AxiomC{$\Gamma \vdash \Delta$} 
\RightLabel{$w \colon l$} 
\UnaryInfC{$A, \Gamma \vdash \Delta$} 
\DisplayProof 
$\qquad$ and $\qquad$ 
\AxiomC{$\Gamma \vdash \Delta$}
\RightLabel{$w \colon r$} 
\UnaryInfC{$\Gamma \vdash \Delta, A$} 
\DisplayProof \newline
\end{center}

\item Contraction rules: \\
\begin{center}
\AxiomC{$A, A, \Gamma \vdash \Delta$} 
\RightLabel{$c \colon l$} 
\UnaryInfC{$A, \Gamma \vdash \Delta$} 
\DisplayProof 
$\qquad$ and $\qquad$ 
\AxiomC{$\Gamma \vdash \Delta, A, A$} 
\RightLabel{$c \colon r$} 
\UnaryInfC{$\Gamma \vdash \Delta, A$} 
\DisplayProof \newline
\end{center}

\item Cut rule:
\begin{center}
\AxiomC{$\Gamma \vdash \Delta, A$} 
\AxiomC{$A, \Pi \vdash \Lambda$} 
\RightLabel{$cut$} 
\BinaryInfC{$\Gamma, \Pi \vdash \Delta, \Lambda$} 
\DisplayProof
\end{center}
\end{itemize}
\end{enumerate}
{\LK}-proofs are endowed with an {\em ancestor relation} on occurrences
of formulas in a natural way (for a definition see~\cite{Baaz2000}).
We now consider two extensions of {\LK}; the first one contains an equality rule which makes the notation of mathematical proofs more practical, the second extension
is by the induction rule.

Let $\Ecal$ be an equational theory (i.e. a finite set of equations). We define the rule
%
%
\[
\AxiomC{$S[t]$} \RightLabel{$\Ecal$} \UnaryInfC{$S[t']$}
  \DisplayProof
\] with the condition that $\Ecal \models t = t'$ and call the corresponding extension of {\LK} {\LKE}. Note that, without restrictions on $\Ecal$, the applicability
of the rule $\Ecal$ is undecidable in general. However, in our paper, the equational theories consist of equations which can be oriented to terminating and confluent
rewrite systems and thus are decidable.

Finally we extend the language of $\LKE$ by the language of arithmetic and write $\omega$ for the sort of natural numbers. The {\em induction rule} is defined as
\begin{center}
\AxiomC{$A(k),\Gamma\seq\Delta,A(k+1)$} \RightLabel{$\ind$} \UnaryInfC{$A(0), \Gamma\seq\Delta, A(t)$} \DisplayProof
\end{center}
where $k$ is a variable of sort $\omega$, $t$ is a term of sort $\omega$, and $k$ does not occur in $\Gamma,\Delta, A(0)$. $A(k)$ is called the {\em induction invariant}.
The resulting calculus is denoted by $\LKIE$.

Consider one of the calculi defined above. A proof is a tree where the nodes are labeled by sequent occurrences and edges are labeled by rules in the usual way. A
proof of $S$ is a proof with root node $S$. Let $\Acal$ be a set of sequents; a proof $\varphi$ of $S$ from $\Acal$ is a proof of $S$ where all leaves of $\varphi$
belong to $\Acal$. If not stated otherwise $\Acal$ is defined as the set of sequents of the form $A \seq A$ for atomic formulas $A$ over the underlying syntax. Atomic
sequents are called {\em clauses}.

A proof $\varphi$ is called {\em cut-free} if the cut rule does not occur in $\varphi$. $\varphi$ is called an ACNF (atomic cut normal form) if all cuts are on atomic
formulas only.

\section{The CERES Method for First-Order Logic}\label{sec:ceres}

The cut-elimination method defined by G. Gentzen in his famous paper~\cite{Gentzen1935} is based on proof rewriting. This rewriting takes place locally (on a
cut-inference in the proof) without taking into account the overall structure of the proof. As a consequence, the method (though elegant and of theoretical
importance) is redundant and inefficient as a tool for analyzing mathematical proofs.

In \cite{Baaz2000,Baaz2006a} the method CERES (Cut-Elimination by RESolution) was defined which takes into account the global structure of a proof $\varphi$ with cut;
this global structure is represented as a quantifier-free formula generally represented as a clause-set term $\Theta(\varphi)$ (evaluating to sets of clauses
$\CL(\varphi)$. It can be shown that $\CL(\varphi)$ is always unsatisfiable. A resolution refutation $\rho$ of $\CL(\varphi)$ then defines a skeleton of an ACNF of
the proof $\varphi$. The final step consists in inserting so-called proof projections into $\rho$ to obtain an ACNF of $\varphi$. The single steps of the method are
illustrated in more detail below.

\begin{df}[clause-set term]\label{def.clsterm}
Clause set terms are binary terms defined as
\begin{itemize}
\item If $C$ is a clause then $[C]$ is a clause-set term.
\item If $t_1$ and $t_2$ are clause-set terms then $t_1 \Plus t_2$ and $t_1 \Times t_2$ are clause-set terms.
\end{itemize}
\end{df}

\begin{df}[semantics of clause-set terms]\label{def.clstermsem}
The mapping $|\mbox{ }|$ maps clause-set terms into sets of clauses by
\begin{itemize}
\item $|[C]| = \{C\}$,
\item $|t_1 \Plus t_2| = |t_1| \union |t_2|$,
\item $|t_1 \Times t_2| = |t_1| \times |t_2|$.
\end{itemize}
\end{df}

The first step of CERES consists in the definition of a clause-set term corresponding to the sub-derivations of an {\LK}-derivation ending in a cut. In particular we
focus on derivations of the cut formulas themselves, i.e. on the derivation of formulas having no successors in the end-sequent.
\begin{df}[characteristic term]\label{def.CHCT}
Let $\varphi$ be an {\LK}-derivation of $S$ and let $\Omega$ be the set of all occurrences of cut formulas in $\varphi$. 
Let $\rho$ be an inference in $\varphi$. We define the clause-set term $\Theta_\rho(\varphi)$ inductively:
  \begin{itemize}
    \item if $\rho$ is an axiom $S'$, let $S''$ be the subsequent of $S'$ consisting of all atoms which are ancestors of an occurrence in $\Omega$, then $\Theta_\rho(\varphi) = [S'']$.
    \item if $\rho$ is a unary rule with immediate predecessor $\rho'$, then $\Theta_\rho(\varphi) = \Theta_{\rho'}(\varphi).$
    \item if $\rho$ is a binary rule with immediate predecessors $\rho_1, \rho_2$, then
      \begin{itemize}
       \item if the auxiliary formulas of $\rho$ are ancestors of $\Omega$, then $\Theta_\rho(\varphi) = \Theta_{\rho_1}(\varphi) \oplus \Theta_{\rho_2}(\varphi)$,
       \item otherwise $\Theta_\rho(\varphi) = \Theta_{\rho_1}(\varphi) \otimes \Theta_{\rho_2}(\varphi).$
      \end{itemize}
  \end{itemize}
%
 Note that, in a binary inference, either both auxiliary formulas
 are ancestors of $\Omega$ or none of them.

 Finally the {\em characteristic term} $\Theta(\vphi)$ is defined
 as $\Theta_{\rho_0}(\vphi)$ where $\rho_0$ is the last inference of $\varphi$.
\end{df}

\begin{df}[characteristic clause set]\label{def.CHCS}
Let $\vphi$ be an {\LK}-derivation and $\Theta(\vphi)$ be the
characteristic term of $\vphi$. Then $\CL(\vphi)$, for $\CL(\vphi)
= |\Theta(\vphi)|$, is called the {\em characteristic clause set}
of $\vphi$.
\end{df}

\begin{example}\label{ex.CL}
Let $\vphi$ be the derivation (for $u,v$ free variables, $a$ a
constant symbol)
\[
  \infer[cut]{(\all x)(\neg P(x) \Or Q(x)) \seq (\ex y)Q(y)}{
    \varphi_1\ \ \ & \ \ \ \varphi_2}
\]
where $\vphi_1$ is the {\LK}-derivation:
\[
\infer[\all:r]{(\all x)(\neg P(x) \Or Q(x)) \seq (\all x)(\ex y)
        (\neg P(x) \Or Q(y))^\star}
  { \infer[\all:l]{(\all x)(\neg P(x) \Or Q(x)) \seq (\ex y)
        (\neg P(u) \Or Q(y))^\star}
   { \infer[\ex:r]{ \neg P(u) \Or Q(u) \seq (\ex y)
            (\neg P(u) \Or Q(y))^\star}
    { \infer[c:r]{ \neg P(u) \Or Q(u) \seq (\neg P(u) \Or Q(u))^\star}
     { \infer[\Or:r]{ \neg P(u) \Or Q(u) \seq (\neg P(u) \Or Q(u))^\star,
            (\neg P(u) \Or Q(u))^\star}
      { \infer[\Or:r]{ \neg P(u) \Or Q(u) \seq (\neg P(u) \Or Q(u))^\star,
                Q(u)^\star}
       { \infer[\neg:r]{ \neg P(u) \Or Q(u) \seq \neg P(u)^\star, Q(u)^\star}
        { \infer[\Or:l]{P(u)^\star, \neg P(u) \Or Q(u) \seq Q(u)^\star}
          { \infer[\neg :r]{\neg P(u), P(u)^\star \seq Q(u)^\star}{P(u)^\star \seq
                    Q(u)^\star,P(u)}
             &
            Q(u), P(u)^\star \seq Q(u)^\star
           }
         }
        }
       }
      }
     }
    }
   }
\]
and $\vphi_2$ is:
\[
\infer[\all:l]{(\all x)(\ex y)(\neg P(x) \Or Q(y))^\star \seq
    (\ex y)Q(y)}
 { \infer[\ex:l]{(\ex y)(\neg P(a) \Or Q(y))^\star \seq
     (\ex y)Q(y)}
   { \infer[\ex:r]{(\neg P(a) \Or Q(v))^\star \seq (\ex y)Q(y)}
     { \infer[\Or:l]{(\neg P(a) \Or Q(v))^\star \seq Q(v)}
       { \infer[\neg:l]{\neg P(a)^\star \seq Q(v)}{\seq Q(v),P(a)^\star}
         &
         Q(v)^\star \seq Q(v)
       }
     }
   }
 }
\]

Let $\Omega$ be the set of the two occurrences of the cut formula in $\vphi$. The ancestors of $\Omega$ are marked by $\star$. We compute the characteristic term
$\Theta(\vphi)$:

From the $\star$-marks in $\vphi$  we first get the clause-set terms corresponding to the initial sequents:
$$X_1 = [P(u) \seq Q(u)],\ X_2 = [P(u) \seq Q(u)],\
  X_3 = [\seq P(a)],\ X_4 = [Q(v) \seq].$$
The leftmost-uppermost inference in $\vphi_1$ is unary and thus
the clause term $X_1$ corresponding to this position does not
change. The first binary inference in $\vphi_1$ (it is $\Or:l$)
takes place on non-ancestors of $\Omega$ -- the auxiliary formulas
of the inference are not marked by $\star$. Consequently we obtain
the term
$$Y_1 = [P(u) \seq Q(u)] \Times [P(u) \seq Q(u)].$$
The following inferences in $\vphi_1$ are all unary and so we
obtain
$$\Theta(\vphi)/\nu_1 = Y_1$$
for $\nu_1$ being the position of the end sequent of $\vphi_1$ in
$\vphi$.

Again the uppermost-leftmost inference in $\vphi_2$ is unary and
thus $X_3$ does not change. The first binary inference in
$\vphi_2$ takes place on ancestors of $\Omega$ (the auxiliary
formulas are $\star$-ed) and we have to apply the $\Plus$ to $X_3,
X_4$. So we get
$$Y_2 = [\seq P(a)] \Plus [Q(v)\seq].$$
 Like in $\vphi_1$ all following inferences in
$\vphi_2$ are unary leaving the clause-set term unchanged. Let $\nu_2$ be the occurrence of the end-sequent of $\vphi_2$ in $\vphi$. Then the corresponding clause
term is
$$\Theta(\vphi)/\nu_2 = Y_2.$$
The last inference (cut) in $\vphi$ takes place on ancestors of
$\Omega$ and we have to apply $\Plus$ again. This eventually
yields the characteristic term
\begin{eqnarray*}
 \Theta(\vphi) &=& Y_1 \Plus Y_2 =\\
 \mbox{} && ([P(u) \seq Q(u)] \Times [P(u) \seq Q(u)]) \Plus
    ([\seq P(a)] \Plus [Q(v)\seq]).
\end{eqnarray*}
For the characteristic clause set we obtain
$$\CL(\vphi) = |\Theta(\vphi)| = \{P(u),P(u) \seq Q(u),Q(u);\
    \seq P(a);\ Q(v) \seq\}.$$
\end{example}

It is easy to verify that the set of characteristic clauses
$\CL(\vphi)$ constructed in the example above is unsatisfiable.
This is not merely a coincidence, but a general principle
expressed in the next proposition.

\begin{proposition}\label{prop.uns}
Let $\vphi$ be an {\LK}-derivation. Then $\CL(\vphi)$ is
unsatisfiable.
\end{proposition}
\begin{proof}
In \cite{Baaz2000}.
\end{proof}

Let $\vphi$ be a deduction of $S\colon \Gamma \seq \Delta$ and $\CL(\vphi)$ be the characteristic clause set of $\vphi$. Then $\CL(\vphi)$ is unsatisfiable and,  by
the completeness of resolution (see \cite{Robinson1965}), there exists a resolution refutation $\gamma$ of $\CL(\vphi)$. By applying a ground projection to $\gamma$
we obtain a ground resolution refutation $\gamma'$ of $\CL(\vphi)$; by our definition of resolution $\gamma'$ is also an AC-deduction of $\seq$ from (ground instances
of) $\CL(\vphi)$. This deduction $\gamma'$ may serve as a skeleton of an AC-deduction $\psi$ of $\Gamma \seq \Delta$ itself. The construction of $\psi$ from $\gamma'$
is based on {\em projections} replacing $\vphi$ by cut-free deductions $\vphi(C)$ of $\Pb,\Gamma \seq \Delta, \Qb$ for clauses $C: \Pb \seq \Qb$ in $\CL(\varphi)$. We
merely give an informal description of the projections, for details we refer to \cite{Baaz2000,Baaz2006a}. Roughly speaking, the projections of the proof $\vphi$ are
obtained by skipping all the inferences leading to a cut. As a ``residue" we obtain a characteristic clause in the end sequent. Thus a projection is a cut-free
derivation of the end sequent $S$ $+$ some atomic formulas in $S$. For the application of projections it is vital to have a skolemized end sequent, otherwise
eigenvariable conditions could be violated.

The construction of $\vphi(C)$ is illustrated below.

\begin{example}\label{ex.proj}
Let $\vphi$ be the proof of the sequent
$$S:(\all x)(\neg P(x) \Or Q(x)) \seq (\ex y)Q(y)$$
as defined in Example~\ref{ex.CL}. We have shown that
$$\CL(\vphi) = \{P(u),P(u) \seq Q(u),Q(u);\ \ \seq P(a);\ \ Q(v) \seq\}.$$
We now define $\vphi(C_1)$, the ``projection" of $\vphi$ to
$C_1\colon P(u),P(u) \seq Q(u),Q(u)$:\\
The problem can be reduced to a projection in $\vphi_1$ because
the last inference in $\vphi$ is a cut and
\begin{eqnarray*}
 \Theta(\vphi)/\nu_1 &=& [P(u) \seq Q(u)] \Times [P(u) \seq Q(u)].
\end{eqnarray*}
By skipping all inferences in $\vphi_1$ leading to the cut
formulas we obtain the deduction
\[
\infer[\all:l]{P(u),P(u), (\all x)(\neg P(x) \Or Q(x)) \seq Q(u),Q(u)}
 { \infer[\Or:l]{P(u),P(u), \neg P(u) \Or Q(u) \seq Q(u),Q(u)}
   { \infer[\neg:l]{\neg P(u), P(u) \seq Q(u)}{P(u)\seq P(u),Q(u)}
     &
     Q(u),P(u) \seq Q(u)
   }
 }
\]

In order to obtain the end sequent we only need an additional
weakening and $\vphi(C_1) =$
\[
\infer[w:r]{P(u),P(u), (\all x)(\neg P(x) \Or Q(x)) \seq (\ex y)Q(y),
        Q(u),Q(u)}
 { \infer[\all:l]{P(u),P(u), (\all x)(\neg P(x) \Or Q(x)) \seq
Q(u),Q(u)}
 { \infer[\Or:l]{P(u),P(u), \neg P(u) \Or Q(u) \seq Q(u),Q(u)}
   { \infer[\neg:l]{\neg P(u), P(u) \seq Q(u)}{P(u)\seq P(u),Q(u)}
     &
     Q(u),P(u) \seq Q(u)
   }
 }
 }
\]
For $C_2 =\ \seq P(a)$ we obtain the projection $\vphi(C_2)$:
\[
\infer[w:l]{(\all x)(\neg P(x) \Or Q(x)) \seq (\ex y)Q(y), P(a)}
 { \infer[\ex:r]{\seq P(a), (\ex y)Q(y)}
  { \seq P(a), Q(v)}
 }
\]
Similarly we obtain $\vphi(C_3)$:
\[
\infer[w:l]{(\all x)(\neg P(x) \Or Q(x)), Q(v) \seq (\ex y)Q(y)}
 { \infer[\ex:r]{Q(v) \seq (\ex y)Q(y)}
   { Q(v) \seq Q(v)}
 }
\]
\end{example}

Let $\vphi$ be a proof of $S$ s.t. $\vphi$ is skolemized and let
$\gamma$ be a resolution refutation of the (unsatisfiable) set of
clauses $\CL(\vphi)$. Then $\gamma$ can be transformed into a
deduction $\vphi(\gamma)$ of $S$ s.t.
$\vphi(\gamma)$ is a proof with atomic cuts, thus an AC-normal
form of $\vphi$. $\vphi(\gamma)$ is constructed from $\gamma$
simply by replacing the resolution steps by the corresponding
proof projections. The construction of $\vphi(\gamma)$ is the
essential part of the method CERES (the final elimination of
atomic cuts is inessential). The resolution refutation $\gamma$
can be considered as the characteristic part of $\vphi(\gamma)$
representing the essential result of AC-normalization. Below we
give an example of a construction of $\vphi(\gamma)$, for details
we refer to \cite{Baaz2000,Baaz2006a} again.

\begin{example}\label{ex.atc}
Let $\vphi$ be the proof of
$$S\colon (\all x)(\neg P(x) \Or Q(x)) \seq (\ex y)Q(y)$$
as defined in Example~\ref{ex.CL} and in Example~\ref{ex.proj}.
Then
$$\CL(\vphi) = \{C_1: P(u),P(u) \seq Q(u),Q(u);\ C_2:\ \seq P(a);\
C_3: Q(v) \seq\}.$$
First we define a resolution refutation $\delta$ of $\CL(\vphi)$:
\[
\infer[R]{\seq}{
    \infer[R]{\seq Q(a),Q(a)}{\seq P(a) & P(u),P(u) \seq Q(u),Q(u)}
    &
    Q(v) \seq
}
\]
and a corresponding ground refutation $\gamma$:
\[
\infer[R]{\seq}{
    \infer[R]{\seq Q(a),Q(a)}{\seq P(a) & P(a),P(a) \seq Q(a),Q(a)}
    &
    Q(a) \seq
}
\]
The ground substitution defining the ground projection is $\sigma: \{u \ass a, v \ass a\}$.

Let $\chi_1 = \vphi(C_1)\sigma$, $\chi_2 = \vphi(C_2)\sigma$ and
$\chi_3 = \vphi(C_3)\sigma$. Moreover let us write $B$ for $(\all
x)(\neg P(x) \Or Q(x))$ and $C$ for $(\ex y)Q(y)$.

Then $\vphi(\gamma)$ is of the form
\[
\infer[c:r^*]{B \seq C}
 { \infer[c:l^*]{B \seq C,C,C}
  { \infer[cut]{B,B,B \seq C,C,C}
   { \infer[cut]{B,B \seq C,C,Q(a)}
    {
       \deduce{B \seq C,P(a)}{(\chi_2)}
         &
       \deduce{P(a),B \seq C,Q(a)}{(\chi_1)}
     }
     &
    \deduce{Q(a),B \seq C}{(\chi_3)}
  }
 }
}
\]
\end{example}

\section{CERES and Induction}\label{sec:motivation}

We now turn our attention to the issue of cut-elimination in the presence of induction.

Let us consider the sequent $S$:
\begin{equation*}
(\forall x) (P(x) \impl P(f(x))) \seq (\forall n)(\forall x) ( (P(\hat{f}(n,x)) \impl P(g(n,x))) \impl (P(x) \impl P(g(n,x))) )
\end{equation*}
where $g$ is a binary function symbol, $f$ is unary one and
$$\Ecal = \{\fhat(0,x) = x, \fhat(s(n),x) = f(\fhat(n,x))\}.$$
Obviously, $S$ cannot be proven without induction, which can be shown via the fact that $S$ does not have a Herbrand sequent (w.r.t. the theory $\Ecal$) . That means
that there exists no proof of $S$ in $\LKE$. In fact we need the following inductive lemma:
\begin{equation*}
(\forall x) (P(x) \impl P(f(x))) \seq (\forall n)(\forall x) (P(x) \impl P(\hat{f}(n,x))).
\end{equation*}

A proof $\psi$ of this inductive lemma in $\LKIE$ could be: 
\begin{prooftree}
\scriptsize
\AxiomC{$(\psi_1)$}
\noLine
\UnaryInfC{$\seq (\forall x) (P(x) \impl P(\hat{f}(\Null,x)))$}

\AxiomC{$(\psi_2)$}
\noLine
\UnaryInfC{$\Gamma, (\forall x) (P(x) \impl P(\hat{f}(\alpha,x))) \seq (\forall x) (P(x) \impl P(\hat{f}(s(\alpha),x)))$}
\RightLabel{$ind$}
\UnaryInfC{$\Gamma, (\forall x) (P(x) \impl P(\hat{f}(\Null,x))) \seq (\forall x) (P(x) \impl P(\hat{f}(\gamma,x)))$}
\RightLabel{$\forall \colon r$}
\UnaryInfC{$\Gamma, (\forall x) (P(x) \impl P(\hat{f}(\Null,x))) \seq (\forall n)(\forall x) (P(x) \impl P(\hat{f}(n,x)))$}

\RightLabel{$cut$}
\BinaryInfC{$(\forall x) (P(x) \impl P(f(x))) \seq (\forall n)(\forall x) (P(x) \impl P(\hat{f}(n,x)))$}
\end{prooftree}

\normalsize where $\Gamma = (\forall x) (P(x) \impl P(f(x)))$. The proofs $\psi_1$ and $\psi_2$ are easily defined; 
$\psi_1$ is:
\begin{prooftree}
\scriptsize
\AxiomC{$P(\hat{f}(\Null,u)) \seq P(\hat{f}(\Null,u))$}
\RightLabel{$\Ecal$}
\UnaryInfC{$P(u) \seq P(\hat{f}(\Null,u))$}
\RightLabel{$\impl \colon r$}
\UnaryInfC{$\seq P(u) \impl P(\hat{f}(\Null,u))$}
\RightLabel{$\forall \colon r$}
\UnaryInfC{$\seq (\forall x) (P(x) \impl P(\hat{f}(\Null,x)))$}
\end{prooftree}
\normalsize and $\psi_2$ is: 
\begin{prooftree}
\scriptsize
\AxiomC{$P(u) \seq P(u)$}

\AxiomC{$P(\hat{f}(\alpha,u)) \seq P(\hat{f}(\alpha,u))$}

\AxiomC{$P(\hat{f}(s(\alpha),u)) \seq P(\hat{f}(s(\alpha),u))$}
\RightLabel{$\Ecal$}
\UnaryInfC{$P(f(\hat{f}(\alpha,u))) \seq P(\hat{f}(s(\alpha),u))$}
\RightLabel{$\impl \colon l$}
\BinaryInfC{$P(\hat{f}(\alpha,u)) \impl P(f(\hat{f}(\alpha,u))), P(\hat{f}(\alpha,u)) \seq P(\hat{f}(s(\alpha),u))$}
\RightLabel{$\forall \colon l$}
\UnaryInfC{$(\forall x) (P(x) \impl P(f(x))), P(\hat{f}(\alpha,u)) \seq P(\hat{f}(s(\alpha),u))$}

\RightLabel{$\impl \colon l$}
\BinaryInfC{$P(u),(\forall x) (P(x) \impl P(f(x))), P(u) \impl P(\hat{f}(\alpha,u)) \seq P(\hat{f}(s(\alpha),u))$}
\RightLabel{$\impl \colon r$}
\UnaryInfC{$(\forall x) (P(x) \impl P(f(x))), P(u) \impl P(\hat{f}(\alpha,u)) \seq P(u) \impl P(\hat{f}(s(\alpha),u))$}
\RightLabel{$\forall \colon l$}
\UnaryInfC{$(\forall x) (P(x) \impl P(f(x))), (\forall x) (P(x) \impl P(\hat{f}(\alpha,x))) \seq P(u) \impl P(\hat{f}(s(\alpha),u)))$}
\RightLabel{$\forall \colon r$}
\UnaryInfC{$(\forall x) (P(x) \impl P(f(x))), (\forall x) (P(x) \impl P(\hat{f}(\alpha,x))) \seq (\forall x) (P(x) \impl P(\hat{f}(s(\alpha),x)))$}
\end{prooftree}

\normalsize
Finally, we define $\varphi$ as (to gain some space, the cut-formula $(\forall n)(\forall x) (P(x) \impl P(\hat{f}(n,x)))$ is denoted with $C$):
\begin{prooftree}
\scriptsize
\AxiomC{$(\psi)$}
\noLine
\UnaryInfC{$(\forall x) (P(x) \impl P(f(x))) \seq C$}

\AxiomC{$(\chi)$}
\noLine
\UnaryInfC{$C \seq (\forall n)(\forall x) ( (P(\hat{f}(n,x)) \impl P(g(n,x))) \impl (P(x) \impl P(g(n,x))) )$}

\RightLabel{$cut$}
\BinaryInfC{$(\forall x) (P(x) \impl P(f(x))) \seq (\forall n)(\forall x) ( (P(\hat{f}(n,x)) \impl P(g(n,x))) \impl (P(x) \impl P(g(n,x))) )$}
\end{prooftree}
\normalsize
where $\chi$ is an induction-free proof of the form:
\begin{prooftree}
 \scriptsize
\AxiomC{$P(u) \seq P(u)$}

\AxiomC{$P(\hat{f}(\beta,u)) \seq P(\hat{f}(\beta,u))$}
\AxiomC{$P(g(\beta,u)) \seq P(g(\beta,u))$}

\RightLabel{$\impl \colon l$}
\BinaryInfC{$P(\hat{f}(\beta,u)) \impl P(g(\beta,u)), P(\hat{f}(\beta,u)) \seq P(g(\beta,u))$}

\RightLabel{$\impl \colon l$}
\BinaryInfC{$P(u), P(\hat{f}(\beta,u)) \impl P(g(\beta,u)), P(u) \impl P(\hat{f}(\beta,u)) \seq P(g(\beta,u))$}
\RightLabel{$\impl \colon r$}
\UnaryInfC{$P(\hat{f}(\beta,u)) \impl P(g(\beta,u)), P(u) \impl P(\hat{f}(\beta,u)) \seq P(u) \impl P(g(\beta,u))$}
\RightLabel{$\impl \colon r$}
\UnaryInfC{$P(u) \impl P(\hat{f}(\beta,u)) \seq (P(\hat{f}(\beta,u)) \impl P(g(\beta,u))) \impl (P(u) \impl P(g(\beta,u)))$}
\RightLabel{$\forall \colon l *$}
\UnaryInfC{$(\forall n)(\forall x) (P(x) \impl P(\hat{f}(n,x))) \seq (P(\hat{f}(\beta,u)) \impl P(g(\beta,u))) \impl (P(u) \impl P(g(\beta,u)))$}
\RightLabel{$\forall \colon r *$}
\UnaryInfC{$(\forall n)(\forall x) (P(x) \impl P(\hat{f}(n,x))) \seq (\forall n)(\forall x) ( (P(\hat{f}(n,x)) \impl P(g(n,x))) \impl (P(x) \impl P(g(n,x))) )$}
\end{prooftree}

\normalsize In the attempt of performing reductive cut-elimination a la Gentzen, we locate the place in the proof, where $(\forall n)$ is introduced. In $\chi$,
$(\forall n)(\forall x) (P(x) \impl P(\hat{f}(n,x)))$ is obtained from $(\forall x) (P(x) \impl P(\hat{f}(\beta,x)))$ by $\forall \colon l$. In the proof $\psi$ we
may delete the $\forall \colon r$ inference yielding the cut-formula and replace $\gamma$ by $\beta$. But in the attempt to eliminate $(\forall x) (P(x) \impl
P(\hat{f}(\beta,x)))$ in $\psi$ we get stuck, as we cannot ``cross'' the $ind$ rule. Neither can the $ind$ rule be eliminated as $\beta$ is variable. In fact, if we
had instead $(\forall x) (P(x) \impl P(\hat{f}(t,x)))$ for a closed term $t$ over $\{\Null,s,+,*\}$ we could prove $\seq t = \bar{n}$ from the axioms of Peano
arithmetic and also
\scriptsize
\begin{equation*}
(\forall x) (P(x) \impl P(f(x))), (\forall x) (P(x) \impl P(\hat{f}(\Null,x))) \seq (\forall x) (P(x) \impl P(\hat{f}(\bar{n},x)))
\end{equation*}
\normalsize
without induction (by iterated cuts) and cut-elimination would proceed.

This problem, however, is neither rooted in the specific form of $\psi$ nor the $ind$ rule. Even if we had used the binary induction rule,
\begin{prooftree}
\AxiomC{$\Gamma \seq \Pi, A(\Null)$} 
\AxiomC{$\Delta, A(\alpha) \seq \Lambda, A(s(\alpha))$} 
\RightLabel{${\rm ind}$} 
\BinaryInfC{$\Gamma, \Delta \seq \Pi, \Lambda, A(t)$}
\end{prooftree}
the result would be the same. In fact, there exists no proof of $S$ with only atomic cuts -- even if {\rm ind} is used. In particular, induction on the formula
$$(\forall n)(\forall x) ( (P(\hat{f}(n,x)) \impl P(g(n,x))) \impl (P(x) \impl P(g(n,x))) )$$ fails. In order to prove the end-sequent an inductive lemma is needed;
something which implies $(\forall n)(\forall x) (P(x) \impl P(\hat{f}(n,x)))$ and cannot be eliminated.

While there are no proof of $S$ in $\LKE$ with the axioms of minimal arithmetic and only atomic cuts, the sequents $S_n$:
\scriptsize
\begin{equation*}
(\forall x) (P(x) \impl P(f(x))) \seq (\forall x) ( (P(\hat{f}(\bar{n},x)) \impl P(g(\bar{n},x))) \impl (P(x) \impl P(g(\bar{n},x))) )
\end{equation*}
\normalsize
do have such proofs in $\LKE$ for all $n$; indeed, they can be proved without induction. But instead of a unique proof $\varphi$ of $S$ we get an infinite
sequence of proofs $\varphi_n$ of $S_n$, which have cut-free versions $\varphi_n'$ (henceforth ``cut-free'' means that atomic cuts are admitted). This kind of
``infinitary'' cut-elimination only makes sense if there exists a uniform representation of the sequence of proofs $\varphi_n'$. We will illustrate below that the
method CERES has the potential of producing such a uniform representation, thus paving the way for cut-elimination in the presence of induction.

Let $S_n$ be
\scriptsize
$$(\forall x) (P(x) \impl P(f(x))) \seq (\forall x) ( (P(\hat{f}(n,x)) \impl P(g(n,x))) \impl (P(x) \impl P(g(n,x))) )$$
\normalsize
where $n$ is a number variable, henceforth called a parameter and $\Ecal$ be the equational theory defined above.

First we define a proof schema $\psi_n$ playing the role of $\psi$ in the inductive proof above; $\psi_0$ is:
\begin{prooftree}
\scriptsize \AxiomC{$P(\hat{f}(\Null,x_0)) \seq P(\hat{f}(\Null,x_0))$} \RightLabel{$\Ecal$} \UnaryInfC{$P(x_0) \seq P(\hat{f}(\Null,x_0))$} \RightLabel{$\impl \colon
r$} \UnaryInfC{$ \seq P(x_0) \impl P(\hat{f}(\Null,x_0))$} \RightLabel{$\forall \colon r$} \UnaryInfC{$ \seq (\forall x) (P(x) \impl P(\hat{f}(\Null,x)))$}
\RightLabel{$w \colon l$} \UnaryInfC{$(\forall x) (P(x) \impl P(f(x))) \seq (\forall x) (P(x) \impl P(\hat{f}(\Null,x)))$}
\end{prooftree}
and $\psi_{k+1}$ is:
\begin{prooftree}
\scriptsize
\AxiomC{$(\psi_k)$}
\dashedLine
\UnaryInfC{$(\forall x) (P(x) \impl P(f(x))) \seq (\forall x) (P(x) \impl P(\hat{f}(k,x)))$}

\AxiomC{$(1)$}

\RightLabel{$cut, c \colon l$}
\BinaryInfC{$(\forall x) (P(x) \impl P(f(x))) \seq (\forall x) (P(x) \impl P(\hat{f}(s(k),x)))$}
\end{prooftree}
\normalsize where $(1)$ is:
\begin{prooftree}
 \scriptsize
\AxiomC{$P(x_{k+1}) \seq P(x_{k+1})$}

\AxiomC{$P(\hat{f}(k,x_{k+1})) \seq P(\hat{f}(k,x_{k+1}))$}

\AxiomC{$P(\hat{f}(s(k),x_{k+1})) \seq P(\hat{f}(s(k),x_{k+1}))$} \RightLabel{$\Ecal$} \UnaryInfC{$P(f(\hat{f}(k,x_{k+1}))) \seq P(\hat{f}(s(k),x_{k+1}))$}

\RightLabel{$\impl \colon l$}
\BinaryInfC{$P(\hat{f}(k,x_{k+1})), P(\hat{f}(k,x_{k+1})) \impl P(f(\hat{f}(k,x_{k+1}))) \seq P(\hat{f}(s(k),x_{k+1}))$}
\RightLabel{$\forall \colon l$}
\UnaryInfC{$P(\hat{f}(k,x_{k+1})), (\forall x) (P(x) \impl P(f(x))) \seq P(\hat{f}(s(k),x_{k+1}))$}

\RightLabel{$\impl \colon l$}
\BinaryInfC{$P(x_{k+1}), P(x_{k+1}) \impl P(\hat{f}(k,x_{k+1})), (\forall x) (P(x) \impl P(f(x))) \seq P(\hat{f}(s(k),x_{k+1}))$}
\RightLabel{$\impl \colon r$}
\UnaryInfC{$P(x_{k+1}) \impl P(\hat{f}(k,x_{k+1})), (\forall x) (P(x) \impl P(f(x))) \seq P(x_{k+1}) \impl P(\hat{f}(s(k),x_{k+1}))$}
\RightLabel{$\forall \colon l$}
\UnaryInfC{$(\forall x) (P(x) \impl P(\hat{f}(k,x))), (\forall x) (P(x) \impl P(f(x))) \seq P(x_{k+1}) \impl P(\hat{f}(s(k),x_{k+1}))$}
\RightLabel{$\forall \colon r$}
\UnaryInfC{$(\forall x) (P(x) \impl P(\hat{f}(k,x))), (\forall x) (P(x) \impl P(f(x))) \seq (\forall x) (P(x) \impl P(\hat{f}(s(k),x)))$}
\end{prooftree}

\normalsize
Instead of $S_n$ we consider the skolemized version $S_n'$:\footnote{Skolemization is vital for CERES but the situation in the inductive proof above
remains the same -- we get $c$ instead of $\beta$ and again the same argument applies that the cut cannot be eliminated.}
$$(\forall x) (P(x) \impl P(f(x))) \seq (P(\hat{f}(n,c)) \impl P(g(n,c))) \impl (P(c) \impl P(g(n,c)))$$
and define $\varphi_n$:
\begin{prooftree}
\scriptsize
\AxiomC{$(\psi_n)$}
\noLine
\UnaryInfC{$(\forall x) (P(x) \impl P(f(x))) \seq C$}

\AxiomC{$(\chi_n)$}
\noLine
\UnaryInfC{$C \seq (P(\hat{f}(n,c)) \impl P(g(n,c))) \impl (P(c) \impl P(g(n,c)))$}

\RightLabel{$cut$}
\BinaryInfC{$(\forall x) (P(x) \impl P(f(x))) \seq (P(\hat{f}(n,c)) \impl P(g(n,c))) \impl (P(c) \impl P(g(n,c)))$}
\end{prooftree}
\normalsize where $C = (\forall x) (P(x) \impl P(\hat{f}(n,x)))$ and $(\chi_n)$ is:
\begin{prooftree}
 \scriptsize
\AxiomC{$P(c) \seq P(c)$}

\AxiomC{$P(\hat{f}(n,c)) \seq P(\hat{f}(n,c))$}
\AxiomC{$P(g(n,c)) \seq P(g(n,c))$}

\RightLabel{$\impl \colon l$}
\BinaryInfC{$P(\hat{f}(n,c)) \impl P(g(n,c)), P(\hat{f}(n,c)) \seq P(g(n,c))$}

\RightLabel{$\impl \colon l$}
\BinaryInfC{$P(c), P(\hat{f}(n,c)) \impl P(g(n,c)), P(c) \impl P(\hat{f}(n,c)) \seq P(g(n,c))$}
\RightLabel{$\impl \colon r$}
\UnaryInfC{$P(\hat{f}(n,c)) \impl P(g(n,c)), P(c) \impl P(\hat{f}(n,c)) \seq P(c) \impl P(g(n,c))$}
\RightLabel{$\impl \colon r$}
\UnaryInfC{$P(c) \impl P(\hat{f}(n,c)) \seq (P(\hat{f}(n,c)) \impl P(g(n,c))) \impl (P(c) \impl P(g(n,c)))$}
\RightLabel{$\forall \colon l$}
\UnaryInfC{$(\forall x) (P(x) \impl P(\hat{f}(n,x))) \seq (P(\hat{f}(n,c)) \impl P(g(n,c))) \impl (P(c) \impl P(g(n,c)))$}
\end{prooftree}

\normalsize
In the next step we determine the characteristic clause set schema of $\varphi_n$, $\CL(\varphi_n)$, inductively: $\CL(\varphi_0) = \{\seq P(c); \;
P(\hat{f}(\Null,c)) \seq\}$ and $\CL(\varphi_{n+1})$ is
\[
\begin{array}{l}
\{ P(\hat{f}(\Null,x_1)) \seq P(\hat{f}(s(\Null),x_1)); \; \cdots ; \; P(\hat{f}(n,x_{n+1})) \seq P(\hat{f}(s(n),x_{n+1})); \\
 \ \ \seq P(c); \; P(\hat{f}(n+1,c)) \seq \}.
\end{array}
\]
Now, via the theory $\Ecal$, the clause $P(\hat{f}(\Null,x_1)) \seq P(\hat{f}(s(\Null),x_1))$ becomes $P(x_1) \seq P(f(x_1))$ and the clause set boils down -- via
subsumption to:
\begin{equation*}
\CL(\varphi_n)' = \{ P(x_1) \seq P(f(x_1)); \; \seq P(c); \; P(\hat{f}(n,c)) \seq \}
\end{equation*}

A sequence of resolution refutations of $\CL(\varphi_n)'$ is given by $\delta_n$:
\begin{prooftree}
\AxiomC{$(\eta_n)$} \noLine \UnaryInfC{$\seq P(\hat{f}(n,c))$}

\AxiomC{$ P(\hat{f}(n,c)) \seq$}

\BinaryInfC{$\seq$}
\end{prooftree}

where $\eta_0$ is $\seq P(c)$ and $\eta_{k+1}$ is:
\[
\infer[\Ecal]{\seq P(\fhat(k+1,c))}
 { \infer[x_{k+1} \goesto \hat{f}(k,c)]{\seq P(f(\fhat(k,c)))}
  { \deduce{\seq P(\fhat(k,c))}{(\eta_k)}
    &
    P(x_{k+1}) \seq P(f(x_{k+1}))
  }
 }
\]

The projections and the ANCFs of the proof sequences will be illustrated in Section~\ref{sec:schemceres} as an application of the formal machinery to be developed in
the sections to come. We see that by using the CERES method we are capable of computing a {\em uniform} description of the infinite sequence of ACNFs. In fact, this
computation of cut-elimination in a schema can be considered as a cut-elimination method in the presence of induction.

\section{Schematic language}\label{sec:schemlang}

In order to give a systematic treatment of cut-elimination in the presence of induction
along the lines of the previous section, we start by defining 
a {\em schematic first-order language}, i.e.~a formal language
that allows the specification of an (infinite) set of first-order formulas by a finite
term.
Towards this, we work in a two-sorted setting with the sort $\omega$, intended to represent the natural numbers,
and the sort $\iota$, intended to represent an arbitrary first-order domain.
Our language consists
of countable sets of {\em variables} of both sorts, and sorted $n$-ary function and
predicate symbols, i.e.~we associate with every $n$-ary function $f$ a tuple
of sorts $(\tau_1,\ldots,\tau_n,\tau)$ with the intended interpretation $f:\tau_1\times\cdots\times\tau_n\to \tau$,
and analogously for predicate symbols.
We additionally assume that our function symbols are partitioned
into {\em constant function symbols}
and {\em defined function symbols}. The
first set will contain the usual uninterpreted function symbols and the second 
will allow primitive
recursively defined functions in the language. 

{\em Terms} are built from variables and function symbols 
in the usual inductive fashion. We assume the constant function symbols
$0:\omega$ and $s:\omega \to \omega$ (zero and successor)
to be present (if $t:\omega$ we will often write $t+1$ instead of $s(t)$).
By $\vars(t)$ we denote the set of variables of a term $t$.

For every defined function symbol $f$, we assume that its type is $\omega\times \tau_1\times\cdots\times\tau_n\to\tau$ (with $n\geq 0$), and we assume given two rewrite rules
\begin{align*}
f(0,x_1,\ldots,x_n) &\to s,\\
f(s(y),x_1,\ldots,x_n) &\to t[f(y,x_1,\ldots,x_n)]
\end{align*}
such that $\vars(s) \subseteq \{x_1,\ldots,x_n\}$ and $\vars(t[f(y,x_1,\ldots,x_n)]) \subseteq \{x_1,\ldots,x_n, y\}$,
and $s,t[x]$ are terms not containing $f$, and
if a defined function symbol $g$ occurs in $s$ or $t[x]$ then $g\prec f$.
We assume that these rewrite rules are primitive recursive, i.e.~that
$\prec$ is irreflexive.


To denote that an expression $t$ rewrites to an expression $s$ (in arbitrarily many steps), we write $t \rw s$.
\begin{example}\label{ex:term}
The usual primitive recursive definition of addition can be represented in our system:
let $+:\omega\times\omega\to \omega$ be a defined function symbol with the rewrite
rules $+(0,x) \to x$ and $+(s(y),x) \to s(+(y,x))$.
\end{example}
Since our terms can represent the primitive recursive functions, we have the following.
\begin{theorem}
The unification problem of terms is undecidable.
\end{theorem}
\begin{proof}
Consider the $s$-terms defined over the arithmetic signature given by $\Null,S$. Then, by the rewrite rules, we obtain a programming language for the primitive
recursive functions. It can be shown that there are universal Turing machines with a halting predicate $T$ ($T(n,m,k)$ iff program nr. $n$ halts on $m$ within $k$
steps) s.t. $T$ is elementary. The elementary functions are just the functions computable by LOOP-2 programs, 
for which there exists a primitive recursive effective
enumeration (see~\cite{Brainerd1974}). I.e. there exists a function $\psi\colon \N \times \N \to \N$ s.t.
$$\psi(n,m) = \mbox{ result of LOOP-2 program nr. }n \mbox{ on input }m.$$
and for all $f \in E^2_1$ ($E^n_k$ are the elementary functions $\N^n \to \N^k$) there exists a primitive recursive
$h$ s.t.
$$\psi(h(n),m) = f(n,m) \mbox{ for all } n,m \in \N.$$
Now we define a function $g\colon \N \times \N \to \N$ by
\begin{eqnarray*}
g(n,k) &=& k+1 \ \ \mbox{if}\ \neg T(n,n,k)\\
       &=& 0 \ \ \mbox{otherwise}.
\end{eqnarray*}
By definition of $g$ we obtain
$$n \in \bar{K} \ioi   \forall k.\neg T(n,n,k) \ioi \forall k. g(n,k) \neq 0$$
where $K$ is the halting problem. As $\psi$ is a primitive recursive and effective enumeration of $E^1_1$ there exists a primitive recursive function $h$ s.t.
$\psi(h(n),k) = g(n,k)$ for all $n,k \in \N$, and so
$$n \in \bar{K} \ioi \forall k.\psi(h(n),k) \neq 0.$$
Now let $f_\psi \in F_s$ the representation of $\psi$, $f_h \in F_s$ that of $h$. Then deciding the unification problems $\psi(\bar{m},y) = \Null$ for number
constants $\bar{m}$ we can also decide the problem
$$f_\psi(f_h(\bar{n}),y) = \Null \mbox{ for } n \in \N$$
and thus obtain a decision procedure for $K$, which obviously does not exists.\qed
\end{proof}

We now turn to the definition of schematic formulas. Analogously to function symbols, we assume
that the predicate symbols are partitioned into {\em constant predicate symbols} and {\em defined predicate
symbols}, assuming as above rewrite rules and an irreflexive order $\prec$ for the latter. 
{\em Formulas} are then built up inductively 
from atoms using $\land,\lor,\neg,\impl,\forall,\exists$ as usual.
To give concise examples, we sometimes use the formulas $\top,\bot$ although they are not official parts
of the language.
In our setting, it is important to clarify how to interpret multiple occurrences of the same bound variable.
For an occurrence of a bound variable $x$, we consider the lowermost (in the bottom-growing formula-tree) quantifier that
binds $x$ to be associated to that occurrence. Consider the following clarifying example.
\begin{example}
Let $P:\omega$
be a defined predicate symbol, $R:\omega \times \iota$ a constant predicate symbol,
and $x:\iota$ a variable.
For the rewrite rules for $P$, take
\begin{align*}
P(0) &\to \top,\\
P(s(n)) & \to \exists x\;( R(n,x) \land P(n)),\\
\end{align*}
Then we have $P(2) \rw \exists x\;( R(2,x) \land \exists x\;(R(1,x) \land \top))$
which is equivalent to (by renaming of bound variables) $\exists x_2\;( R(2,x_2) \land \exists x_1\;(R(1,x_1) \land \top))$.
\end{example}
%
%
\begin{proposition}
Let $A$ be a formula. Then every rewrite sequence starting at $A$ terminates,
and $A$ has a unique normal form.
\end{proposition}

\begin{proof}
 Trivial, since all definitions are primitive recursive. \qed
\end{proof}
Finally, we note that formulas $\bigvee_{i=s}^t A(i)$ (which rewrite to $\bot$ if $s>t$
and to $A(t)\lor \bigvee_{i=s}^{t-1} A(i)$ otherwise) can be expressed using defined
predicate symbols $P_{A,s,t}$ (see~\cite{Aravantinos2009,Aravantinos2011}), so we will also use them freely in
examples.

\section{Schematic Proofs}\label{sec:schemproof}
We will now give a natural notion of {\em proof schema} for the language defined in the previous section,
and compare these proof schemata with the calculus $\LKIE$.
To this end, we need some notions:
If we introduce a sequent as $S(x_1,\ldots,x_\alpha)$, then by $S(t_1,\ldots,t_\alpha)$ we denote
$S(x_1,\ldots,x_\alpha)$ where $x_1,\ldots,x_\alpha$ are
replaced by $t_1,\ldots,t_\alpha$ respectively, where $t_\beta$ are terms
of appropriate type.
We assume a countably infinite set of {\em proof symbols} denoted by $\varphi, \psi, \ldots$.
If $\varphi$ is a proof symbol and
$S(x_1,\ldots,x_\alpha)$ a sequent, then the expression 
%
\AxiomC{$(\varphi(a_1,\ldots,a_\alpha))$}
\dashedLine
\UnaryInfC{$S(a_1,\ldots,a_\alpha)$}
\DisplayProof
%
is called a {\em proof link}. For a variable $k:\omega$, proof links
such that $\vars(a_1)\subseteq \{k\}$ are called {\em $k$-proof links}.
\begin{definition}
The sequent calculus {\LKS}
consists of the rules of {\LKE}, where proof links may appear
at the leaves of a proof, and where $\mathcal{E}$ is the
set of rewrite rules (interpreted as equations) for the defined
function and predicate symbols.
\end{definition}
\begin{df}[Proof schemata]\label{def.proofschema}
  Let $\psi$ be a proof symbol and $S(n,x_1,\ldots,x_\alpha)$ be a sequent
  such that $n:\omega$.
  Then a {\em proof schema pair for $\psi$} is a pair of {\LKS}-proofs
  $(\pi,\nu(k))$ with end-sequents $S(0,x_1,\ldots,x_\alpha)$ and $S(k+1,x_1,\ldots,x_\alpha)$ respectively
  such that $\pi$ may not contain proof links and $\nu(k)$ may
  contain only proof links of the form
  \AxiomC{$(\psi(k,a_1,\ldots,a_\alpha))$}
  \dashedLine
  \UnaryInfC{$S(k,a_1,\ldots,a_\alpha)$}
  \DisplayProof ($a_1,\ldots,a_\alpha$ terms of appropriate type). For such a proof schema pair, we say that a proof link
  of the form
\AxiomC{$(\psi(a,a_1,\ldots,a_\alpha))$}
  \dashedLine
  \UnaryInfC{$S(a,a_1,\ldots,a_\alpha)$}
  \DisplayProof is a {\em proof link to $\psi$}.
  We say that $S(n,x_1,\ldots,x_\alpha)$ is the end-sequent of $\psi$, and
  we assume an identification between formula occurrences in the end-sequents
  of $\pi$ and $\nu(k)$ so that we can speak of occurrences in the end-sequent
  of $\psi$.

  Finally, a {\em proof schema $\Psi$} is a tuple of proof schema pairs
  $\langle p_1,\ldots,p_\alpha \rangle$ for $\psi_1,\ldots,\psi_\alpha$ respectively
  such that the {\LKS}-proofs in $p_\beta$ may also contain $k$-proof
  links to $\psi_\gamma$ for $1\leq \beta < \gamma \leq \alpha$.
  We also say that the end-sequent of $\psi_1$ is the {\em end-sequent} of $\Psi$.
\end{df}
\begin{example}\label{running_example}
We now present the proof sequence given in Section~\ref{sec:motivation} according to our formal definitions. 
For the defined function symbol $\hat{f}$, we assume as rewrite rules (oriented versions of) the equalities
given in Section~\ref{sec:motivation}.
Then we define a proof schema $\Psi = \langle (\pi_1,\nu_1(k)), (\pi_2,\nu_2(k)) \rangle$ for $\varphi,\psi$.
where $\nu_1(k)$ is defined as:
\begin{prooftree}
\scriptsize
\AxiomC{$(\psi(k+1))$}
\dashedLine
\UnaryInfC{$(\forall x) (P(x) \impl P(f(x))) \seq (\forall x) (P(x) \impl P(\hat{f}(k+1,x)))$}

\AxiomC{$(2)$}

\RightLabel{$cut$}
\BinaryInfC{$(\forall x) (P(x) \impl P(f(x))) \seq (P(\hat{f}(k+1,c)) \impl P(g(k+1,c))) \impl (P(c) \impl P(g(k+1,c)))$}
\end{prooftree}
where $(2)$ is:
\begin{prooftree}
 \scriptsize
\AxiomC{$P(c) \seq P(c)$}

\AxiomC{$P(\hat{f}(k+1,c)) \seq P(\hat{f}(k+1,c))$}
\AxiomC{$P(g(k+1,c)) \seq P(g(k+1,c))$}

\RightLabel{$\impl \colon l$}
\BinaryInfC{$P(\hat{f}(k+1,c)) \impl P(g(k+1,c)), P(\hat{f}(k+1,c)) \seq P(g(k+1,c))$}

\RightLabel{$\impl \colon l$}
\BinaryInfC{$P(c), P(\hat{f}(k+1,c)) \impl P(g(k+1,c)), P(c) \impl P(\hat{f}(k+1,c)) \seq P(g(k+1,c))$}
\RightLabel{$\impl \colon r$}
\UnaryInfC{$P(\hat{f}(k+1,c)) \impl P(g(k+1,c)), P(c) \impl P(\hat{f}(k+1,c)) \seq P(c) \impl P(g(k+1,c))$}
\RightLabel{$\impl \colon r$}
\UnaryInfC{$P(c) \impl P(\hat{f}(k+1,c)) \seq (P(\hat{f}(k+1,c)) \impl P(g(k+1,c))) \impl (P(c) \impl P(g(k+1,c)))$}
\RightLabel{$\forall \colon l$}
\UnaryInfC{$(\forall x) (P(x) \impl P(\hat{f}(k+1,x))) \seq (P(\hat{f}(k+1,c)) \impl P(g(k+1,c))) \impl (P(c) \impl P(g(k+1,c)))$}
\end{prooftree}
Note that the proof link in $\nu_1(k)$ is a $k$-proof link to $\psi$.
For $\pi_1$, we take $\nu_1(k)$ where $k+1$ is replaced by $0$. Hence the end-sequent
of $\varphi$ is $(\forall x) (P(x) \impl P(f(x))) \seq (P(\hat{f}(n,c)) \impl P(g(n,c))) \impl (P(c) \impl P(g(n,c)))$.

We continue giving the definitions for the proof symbol $\psi$. $\pi_2$ is
\begin{prooftree}
\scriptsize
\AxiomC{$P(\hat{f}(0,x_0)) \seq P(\hat{f}(0,x_0))$}
\RightLabel{$\Ecal$}
\UnaryInfC{$P(x_0) \seq P(\hat{f}(0,x_0))$}
\RightLabel{$\impl \colon r$}
\UnaryInfC{$ \seq P(x_0) \impl P(\hat{f}(0,x_0))$}
\RightLabel{$\forall \colon r$}
\UnaryInfC{$ \seq (\forall x) (P(x) \impl P(\hat{f}(0,x)))$}
\RightLabel{$w \colon l$}
\UnaryInfC{$(\forall x) (P(x) \impl P(f(x))) \seq (\forall x) (P(x) \impl P(\hat{f}(0,x)))$}
\end{prooftree}
and $\nu_2(k)$ is
\begin{prooftree}
\scriptsize
\AxiomC{$(\psi(k))$}
\dashedLine
\UnaryInfC{$(\forall x) (P(x) \impl P(f(x))) \seq (\forall x) (P(x) \impl P(\hat{f}(k,x)))$}

\AxiomC{$(1)$}

\RightLabel{$cut, c \colon l$}
\BinaryInfC{$(\forall x) (P(x) \impl P(f(x))) \seq (\forall x) (P(x) \impl P(\hat{f}(k+1,x)))$}
\end{prooftree}
where $(1)$ is:
\begin{prooftree}
 \scriptsize
\AxiomC{$P(x_{k+1}) \seq P(x_{k+1})$}

\AxiomC{$P(\hat{f}(k,x_{k+1})) \seq P(\hat{f}(k,x_{k+1}))$}

\AxiomC{$P(\hat{f}(k+1,x_{k+1})) \seq P(\hat{f}(k+1,x_{k+1}))$}
\RightLabel{$\Ecal$}
\UnaryInfC{$P(f(\hat{f}(k,x_{k+1}))) \seq P(\hat{f}(k+1,x_{k+1}))$}

\RightLabel{$\impl \colon l$}
\BinaryInfC{$P(\hat{f}(k,x_{k+1})), P(\hat{f}(k,x_{k+1})) \impl P(f(\hat{f}(k,x_{k+1}))) \seq P(\hat{f}(k+1,x_{k+1}))$}
\RightLabel{$\forall \colon l$}
\UnaryInfC{$P(\hat{f}(k,x_{k+1})), (\forall x) (P(x) \impl P(f(x))) \seq P(\hat{f}(k+1,x_{k+1}))$}

\RightLabel{$\impl \colon l$}
\BinaryInfC{$P(x_{k+1}), P(x_{k+1}) \impl P(\hat{f}(k,x_{k+1})), (\forall x) (P(x) \impl P(f(x))) \seq P(\hat{f}(k+1,x_{k+1}))$}
\RightLabel{$\impl \colon r$}
\UnaryInfC{$P(x_{k+1}) \impl P(\hat{f}(k,x_{k+1})), (\forall x) (P(x) \impl P(f(x))) \seq P(x_{k+1}) \impl P(\hat{f}(k+1,x_{k+1}))$}
\RightLabel{$\forall \colon l$}
\UnaryInfC{$(\forall x) (P(x) \impl P(\hat{f}(k,x))), (\forall x) (P(x) \impl P(f(x))) \seq P(x_{k+1}) \impl P(\hat{f}(k+1,x_{k+1}))$}
\RightLabel{$\forall \colon r$}
\UnaryInfC{$(\forall x) (P(x) \impl P(\hat{f}(k,x))), (\forall x) (P(x) \impl P(f(x))) \seq (\forall x) (P(x) \impl P(\hat{f}(k+1,x)))$}
\end{prooftree}
Hence the end-sequent of $\psi$ is $(\forall x) (P(x) \impl P(f(x))) \seq (\forall x) (P(x) \impl P(\hat{f}(n,x)))$.
\end{example}

For the rest of this paper, we consider a fixed proof schema
\[
  \langle (\pi_1,\nu_1(k),\ldots,(\pi_\alpha,\nu_\alpha(k)) \rangle \textrm{ for } \psi_1,\ldots,\psi_\alpha.
\]
Proof schemata naturally represent infinite sequences of (first-order) proofs.
We will state this fact formally as a soundness result. First, we define the notion of evaluation of proof schemata.

\begin{df}[Evaluation of proof schemata]
We define the rewrite rules for proof links
\[
\AxiomC{$(\psi(0,x_1,\ldots,x_\beta))$}
\dashedLine
\UnaryInfC{$S(0,x_1,\ldots,x_\beta)$}
\DisplayProof
\rightarrow
\pi,
\qquad
\AxiomC{$(\psi(s(k),x_1,\ldots,x_\beta))$}
\dashedLine
\UnaryInfC{$S(s(k),x_1,\ldots,x_\beta)$}
\DisplayProof
\rightarrow
\nu(k)
\]
for all proof schema pairs $(\pi,\nu(k))$ for $\psi$. 
Now
for $\gamma\in\nats$ we define
$\psi\!\downarrow_\gamma$ as a normal form
of $\AxiomC{$(\psi(\gamma),x_1,\ldots,x_\beta)$}
\dashedLine
\UnaryInfC{$S(\gamma,x_1,\ldots,x_\beta)$}\DisplayProof$ under the rewrite system 
just given extended with rewrite rules for defined function and predicate symbols.
Further, we define $\Psi\!\downarrow_\gamma=\psi_1\!\downarrow_\gamma$.
\end{df}

\begin{proposition}[Soundness of proof schemata]\label{prop:proof_schemata_sound}
  Let $\Psi$ be a proof schema with end-sequent $S(n,x_1,\ldots,x_\beta)$, and let $\gamma\in\nats$.
  Then there exists a first-order proof of $S(\gamma,x_1,\ldots,x_\beta)\!\downarrow$.
\end{proposition}
\begin{proof}
 First we prove the proposition for a proof schema consisting with one pair only and then extend the result to arbitrary proof schemata. Assume $\Psi = \langle (\pi, \nu(k)) \rangle$. We proceed by induction on $\gamma$. If $\gamma = 0$, $\Psi\!\downarrow_0 = \pi\!\downarrow_0$. The later one differs from $\pi$ only in defined function and predicate symbols, therefore $\pi\!\downarrow_0$ is a proof of $S(0,x_1,\ldots,x_\beta)\!\downarrow$. Now assume for all $\delta \leq \gamma$, $\Psi\!\downarrow_\delta$ is a proof of $S(\delta,x_1,\ldots,x_\beta)\!\downarrow$ and consider the case for $\gamma+1$. If $(\psi(k,x_1,\ldots,x_\beta))$ is a proof link in $\nu(k)$, then by hypothesis it rewrites to $\Psi\!\downarrow_\gamma$. Then after applying rewrite rules of defined function and predicate symbols to $\nu(\gamma)$, we get a proof of $S(\gamma+1,x_1,\ldots,x_\beta)\!\downarrow$.

The result is easily extended to arbitrary proof schema $\Psi$, considering the fact that for all $\gamma \in \nats$ each pair $(\pi_i, \nu_i(k)) \in \Psi$ is evaluated to a proof of the sequent $S_i(\gamma,x_1,\ldots,x_\epsilon)\!\downarrow$. \qed
\end{proof}
If $k:\omega$ then an {\LKSI}-proof $\pi$ is called {\em $k$-simple} if
all induction rules in $\pi$ are of the following form:
\begin{center}
\AxiomC{$A(k),\Gamma\seq\Delta,A(k+1)$} \RightLabel{$\ind$} \UnaryInfC{$A(0), \Gamma\seq\Delta, A(t)$} \DisplayProof
\end{center}
such that $\vars(t)\subseteq\{k\}$ (i.e.~$k$ is the ``eigenvariable'' of the induction
inference, and $k$ is the only variable occuring in $t$).
Our next aim is to show that proof schemata and $\LKSI$, restricted to $k$-simple
proofs, are equivalent:
\begin{proposition}\label{prop:schema_to_lksi}
  Let $\Psi$ be a proof schema with end-sequent $S$. Then there exists a $k$-simple {\LKSI}-proof
  of $S$.
\end{proposition}
\begin{proof}
  Let $(\pi_\gamma,\nu_\gamma(k))$ be proofs of $S_\gamma(k)$ 
  for $\psi_\gamma$
  respectively.
  We construct inductively $k$-simple {\LKSI}-proofs of $S_\gamma(k)$, starting with $(\pi_\alpha,\nu_\alpha(k))$.
  By replacing proof links
  \AxiomC{$(\psi_\alpha(k))$}
  \dashedLine
  \UnaryInfC{$S_\alpha(k)$}
  \DisplayProof
  in $\nu_\alpha(k)$ by axioms $S_\alpha(k)\seq S_\alpha(k)$ and using cuts on some easily constructed $\LKE$-proofs
  we obtain an $\LKE$-proof $\lambda$ of $S_\alpha(k)\seq S_\alpha(k+1)$.
  Then the following is the desired $k$-simple {\LKSI}-proof of $S_\alpha(n)$:
  \begin{center}
    \AxiomC{$(\pi_\alpha)$}
    \noLine
    \UnaryInfC{$\seq S_\alpha(0)$}

    \AxiomC{$(\lambda)$}
    \noLine
    \UnaryInfC{$S_\alpha(k)\seq S_\alpha(k+1)$}
    \RightLabel{$\ind$}
    \UnaryInfC{$S_\alpha(0) \seq S_\alpha(k)$}

    \RightLabel{$cut$}
    \BinaryInfC{$\seq S_\alpha(k)$}
    \DisplayProof
  \end{center}
  For the induction step, assume that we have {\LKSI}-proofs $\lambda_{\gamma+1},\ldots,\lambda_\alpha$ 
  of $S_{\gamma+1}(k), \\ \ldots, S_\alpha(k)$ respectively. Our aim
  is to construct an {\LKSI}-proof of $S_\gamma(k)$. As before, in $\nu_\gamma(k)$ we replace proofs links 
  of the form 
  \AxiomC{$(\psi_\iota(t))$}
  \dashedLine
  \UnaryInfC{$S_\iota(t)$}
  \DisplayProof
  by axioms to obtain an {\LKE}-proof of 
  \[
    S_\alpha(t^\alpha_1),\ldots,S_\alpha(t^\alpha_{\delta_\alpha}),\ldots, S_{\gamma+1}(t^{\gamma+1}_1),\ldots,S_{\gamma+1}(t^{\gamma+1}_{\delta_{\gamma+1}}), S_\gamma(k) \seq S_\gamma(k+1).
  \]
  Substituting for $k$ in $\lambda_{\gamma+1},\ldots,\lambda_\alpha$, we obtain appropriate proofs
  to obtain, via cut, an {\LKSI}-proof of $S_\gamma(k)\seq S_\gamma(k+1)$. Note that the 
  $\lambda_\iota$ remain $k$-simple since the proof links in $\nu_\gamma(k)$ are $k$-proof links.
  We perform an analogous transformation
  on $\pi_\gamma$ to obtain an {\LKSI}-proof of $\seq S_\gamma(0)$. From these proofs we construct the
  desired {\LKSI}-proof using the $\ind$ rule as in the base case. Clearly this application of the $\ind$
  rule is again $k$-simple.
  \qed
\end{proof}
We illustrate this by means of a simple example.
\begin{example}
Consider the proof schema $\Psi=\langle(\pi,\nu(k))\rangle$ for the proof symbol $\psi$
with $\pi=$
\begin{center}
\scriptsize
\AxiomC{$P(0)\seq P(0)$}
\AxiomC{$P(1)\seq P(1)$}
\RightLabel{$\impl:l, \Ecal$}
\BinaryInfC{$P(0),\bigwedge_{i=0}^0(P(i) \impl P(i+1))\seq P(1)$}
\DisplayProof
\end{center}
and $\nu(k)=$
\begin{center}
\scriptsize
  \AxiomC{$(\psi(k))$}
  \dashedLine
  \UnaryInfC{$P(0),\bigwedge_{i=0}^{k}(P(i) \impl P(i+1))\seq P(k+1)$}
  \AxiomC{$P(k+2)\seq P(k+2)$}
  \RightLabel{$\impl:l$}
\BinaryInfC{$P(0),\bigwedge_{i=0}^{k}(P(i) \impl P(i+1)), P(k+1) \impl P(k+2)\seq P(k+2)]$}
\RightLabel{$\land:l$}
\UnaryInfC{$P(0),\bigwedge_{i=0}^{k}(P(i) \impl P(i+1))\land P(k+1) \impl P(k+2)\seq P(k+2)$}
\RightLabel{$\Ecal$}
\UnaryInfC{$P(0),\bigwedge_{i=0}^{k+1}(P(i) \impl P(i+1))\seq P(k+2)$}
\DisplayProof
\end{center}
It is translated to the $\LKSI$-proof of $S(k)=P(0)\land\bigwedge_{i=0}^{k}(P(i) \impl P(i+1))\impl P(k+1)$:
\begin{center}
\scriptsize
\AxiomC{$(\pi)$}
\noLine
\UnaryInfC{$\seq S(0)$}

\AxiomC{canonical proof}
\UnaryInfC{$S(k),P(0),\bigwedge_{i=0}^{k}(P(i) \impl P(i+1))\seq P(k+1)$}

\AxiomC{$P(k+2)\seq P(k+2)$}
\RightLabel{$\impl:l$}
\BinaryInfC{$S(k),P(0),\bigwedge_{i=0}^{k}(P(i) \impl P(i+1)),P(k+1)\impl P(k+2)\seq P(k+2)$}
\RightLabel{$\impl:r,\land:l,\Ecal$}
\UnaryInfC{$S(k) \seq S(k+1)$}
\RightLabel{$\ind$}
\UnaryInfC{$S(0) \seq S(k)$}
\RightLabel{$cut$}
\BinaryInfC{$\seq S(k)$}
\DisplayProof
\end{center}
\end{example}
Towards proving the converse of Proposition~\ref{prop:schema_to_lksi}, we introduce
the calculus $\LKSI'$, which is just $\LKSI$ where the $\ind$ rule is replaced by
a binary version
\begin{center}
\AxiomC{$\Gamma\seq\Delta, A(0)$}
\AxiomC{$A_\beta(k),\Gamma_\beta\seq\Delta_\beta,A_\beta(k+1)$}
\RightLabel{$\ind'$}
\BinaryInfC{$\Gamma_\beta\seq\Delta_\beta, A_\beta(t)$}
\DisplayProof
\end{center}
where again $k$ does not occur in $\Gamma,\Delta,A(0)$. We define $k$-simple
$\LKSI'$-proofs analogously to $k$-simple $\LKSI$-proofs. The following
result is easy to prove.
\begin{proposition}\label{prop:lksi_lksip}
There exists a ($k$-simple) $\LKSI$-proof of $S$ if and only if there
exists a ($k$-simple) $\LKSI'$-proof of $S$.
\end{proposition}
\begin{proposition}
  Let $\pi$ be a $k$-simple {\LKSI}-proof of $S$. 
  Then there exists a proof schema with end-sequent $S$.
\end{proposition}
\begin{proof}
By Proposition~\ref{prop:lksi_lksip} we may assume that $\pi$ is a
$k$-simple $\LKSI'$-proof.
Let $\pi$ contain $\alpha$ induction inferences
\begin{center}
\AxiomC{$\Gamma_\beta\seq\Delta_\beta, A_\beta(0)$}
\AxiomC{$A_\beta(k),\Gamma_\beta\seq\Delta_\beta,A_\beta(k+1)$}
\RightLabel{$\ind'$}
\BinaryInfC{$\Gamma_\beta\seq\Delta_\beta, A_\beta(t)$}
\DisplayProof
\end{center}
where $1\leq \beta \leq \alpha$. W.l.o.g. we assume that if $\gamma<\beta$ then the
induction inference with conclusion $\Gamma_\beta\seq\Delta_\beta, A_\beta(t)$
is above the induction inference with conclusion $\Gamma_\gamma\seq\Delta_\gamma, A_\gamma(t)$
in $\pi$. Further, let $T$ be the transformation taking an $\LKSI$-proof $\lambda$
to an {\LKS}-proof by replacing the induction inferences with
conclusion $\Gamma_\beta\seq\Delta_\beta, A_\beta(t)$
by proofs
\begin{center}
\AxiomC{$(\psi_\beta(t))$}
\dashedLine
\UnaryInfC{$\Gamma_\beta\seq\Delta_\beta,A_\beta(t)$}
\DisplayProof.
\end{center}
Clearly, if $\lambda$ is a $k$-simple proof,
then $T(\lambda)$ will only contain $k$-proof links.

We will inductively
construct a proof schema $\Psi=\langle (\pi_1,\nu_1(k)),\ldots,(\pi_\alpha,\nu_\alpha(k)) \rangle$
where $(\pi_\beta,\nu_\beta(k))$ are proof schema pairs with end-sequent $\Gamma_\beta\seq\Delta_\beta,A_\beta(n)$ for 
proof symbols $\psi_\beta$ respectively. Assume that we have already constructed such proofs for
$\psi_{\beta+1},\ldots,\psi_\alpha$, and consider the induction inference with conclusion
$\Gamma_\beta\seq\Delta_\beta, A_\beta(t)$. Let its left premise be $\lambda_1$
and its right premise $\lambda_2$. We set $\pi_\beta=T(\lambda_1)$, which by definition of $T$ fulfills
the requirements on proof links. 
Further, let $\nu_\beta(k)$ be the proof
\begin{center}
\AxiomC{$(\psi_\beta(k))$}
\dashedLine
\UnaryInfC{$\Gamma_\beta\seq\Delta_\beta,A_\beta(k)$}

\AxiomC{$(T(\lambda_2))$}
\noLine
\UnaryInfC{$A_\beta(k),\Gamma_\beta\seq\Delta_\beta,A_\beta(k+1)$}

\RightLabel{$cut$}
\BinaryInfC{$\Gamma_\beta\seq\Delta_\beta,A_\beta(k+1)$}
\DisplayProof
\end{center}
which also clearly satisfies the requirement on proof links. Summarizing,
$(\pi_\beta,\nu_\beta(k))$ is a proof schema pair with end-sequent $\Gamma_\beta\seq\Delta_\beta,A_\beta(n)$,
as desired.
\qed
\end{proof}

%

\subsection{Discussion of Regularization of Schematic Proofs}

For an ordinary first-order CERES, the characteristic clause set is computed from a regular proof, otherwise it may be satisfiable. Regularization is vital for CERES only in some cases when two different eigenvariables come from different branches of a binary rule, that produces an ancestor of some formula in the end-sequent. Therefore we need also the notion of regularization in proof schemata. For \LKS-proofs we use the usual notion of regularization, i.e. an \LKS-proof is regular iff all eigenvariables are distinct. But this is not enough and the reason is illustrated by the following example: Let $\Psi = \langle (\pi, \nu(k)) \rangle$ where $\pi$ is
\begin{prooftree}
 \scriptsize
 \AxiomC{$P(0,u) \seq P(0,u)$}
 \RightLabel{$\forall \colon l$}
 \UnaryInfC{$(\forall x) P(0,x) \seq P(0,u)$}
 \RightLabel{$\forall \colon r$}
 \UnaryInfC{$(\forall x) P(0,x) \seq (\forall x) P(0,x)$}
\end{prooftree}
and $\nu(k)$ is
\begin{prooftree}
\scriptsize
 \AxiomC{$(\psi, k)$}
 \dashedLine
 \UnaryInfC{$\bigor_{i=0}^k (\forall x) P(i,x) \seq \bigor_{i=0}^k (\forall x) P(i,x)$}

 \AxiomC{$P(k+1,u) \seq P(k+1,u)$}
 \RightLabel{$\forall \colon l$}
 \UnaryInfC{$(\forall x) P(k+1,x) \seq P(k+1,u)$}
 \RightLabel{$\forall \colon r$}
 \UnaryInfC{$(\forall x) P(k+1,x) \seq (\forall x) P(k+1,x)$}

 \RightLabel{$\lor \colon l$}
 \BinaryInfC{$\bigor_{i=0}^{k+1} (\forall x) P(i,x) \seq \bigor_{i=0}^k (\forall x) P(i,x), (\forall x) P(k+1,x)$}
 \RightLabel{$\lor \colon r$}
 \UnaryInfC{$\bigor_{i=0}^{k+1} (\forall x) P(i,x) \seq \bigor_{i=0}^{k+1} (\forall x) P(i,x)$}
\end{prooftree}

Then clearly, $u$ is an eigenvariable, $\pi$ and $\nu(k)$ are regular, but when an instance of the schema $\Psi$ is computed for some $\alpha \not = 0$, the instance is not regular anymore.

To avoid such collisions of eigenvariables, a stronger notion of variable is needed. We introduce {\em variable function symbols} of type $\omega \to \iota$. Then the {\em second order (or schematic) variables} are built from variable function symbols and terms of type $\omega$. The set of such variables is denoted with $\Vso$. The semantics is that if $x \in \Vso$ then for all $\alpha \in \nats$, $x(0), \ldots, x(\alpha)$ corresponds to the sequence of first-order variables $x_0,\ldots,x_\alpha$.

We redefine our notions of term, formula and the like, in the usual inductive fashion, taking into account schematic variables. For example, if $x \in \Vso$ and $f$ is a defined function symbol with the rewrite rules: $f(0,x) \ra x(0)$ and $f(k+1,x) \ra g(f(k,x),x(k+1))$, then $f(n,x)$ is a term, and the sequence of terms for $n = 0,1,2,\ldots$ is $x(0), g(x(0),x(1)), g(g(x(0),x(1)),x(2)), \ldots$.

Another issue is to distinguish between {\em global} and {\em local} eigenvariables. An eigenvariable is \emph{global} if it is propagated through proof links, otherwise it is \emph{local}. This distinction is motivated by the fact that a global eigenvariable must occur in (at least) two different proof schema pairs, where in one it is just a free variable and in the other it is an eigenvariable. Finally, we define the procedure of regularization:

\begin{definition}[Regularization]
Let $\Psi$ be a proof schema. For each pair $(\pi_i, \nu_i(k))$ in $\Psi$ we do the following: replace every eigenvariable $x_j$ with $x_j(0)$ in $\pi_i$, every local eigenvariable $x_j$ in $\nu_i(k)$ with $x_j(k+1)$ and every global eigenvariable $y_j$ in $\nu_i(k)$ with $y_j(0)$.
\end{definition}

According to this definition, the regularized version $\Psi'$ of the proof schema $\Psi$, given above, is: $\pi'$ is
\begin{prooftree}
 \scriptsize
 \AxiomC{$P(0,u(0)) \seq P(0,u(0))$}
 \RightLabel{$\forall \colon l$}
 \UnaryInfC{$(\forall x) P(0,x) \seq P(0,u(0))$}
 \RightLabel{$\forall \colon r$}
 \UnaryInfC{$(\forall x) P(0,x) \seq (\forall x) P(0,x)$}
\end{prooftree}
and $\nu(k)$ is
\begin{prooftree}
\scriptsize
 \AxiomC{$(\psi, k)$}
 \dashedLine
 \UnaryInfC{$\bigor_{i=0}^k (\forall x) P(i,x) \seq \bigor_{i=0}^k (\forall x) P(i,x)$}

 \AxiomC{$P(k+1,u(k+1)) \seq P(k+1,u(k+1))$}
 \RightLabel{$\forall \colon l$}
 \UnaryInfC{$(\forall x) P(k+1,x) \seq P(k+1,u(k+1))$}
 \RightLabel{$\forall \colon r$}
 \UnaryInfC{$(\forall x) P(k+1,x) \seq (\forall x) P(k+1,x)$}

 \RightLabel{$\lor \colon l$}
 \BinaryInfC{$\bigor_{i=0}^{k+1} (\forall x) P(i,x) \seq \bigor_{i=0}^k (\forall x) P(i,x), (\forall x) P(k+1,x)$}
 \RightLabel{$\lor \colon r$}
 \UnaryInfC{$\bigor_{i=0}^{k+1} (\forall x) P(i,x) \seq \bigor_{i=0}^{k+1} (\forall x) P(i,x)$}
\end{prooftree}

In the future we consider only regular proof schemata.

\section{Resolution Schemata}\label{sec:resschema}

In this section we define a notion of schematic resolution. In fact, schematic resolution refutations of $\CL(\varPsi)$, combined with the schematic projections
$PR(\varPsi)$ allow the construction of schematic atomic cut normal forms of the original proof schema $\varPsi$ -- what is precisely the aim of a schematic
CERES-method.

\begin{definition}[clause]
Let $p_1,\ldots,p_\alpha$ and $q_1,\ldots,q_\beta$ be schematic atomic formulas; then $p_1,\ldots,p_\alpha \seq q_1,\ldots,q_\beta$ is called a {\em clause}. A clause
in $k$, for an arithmetic variable $k$, is a clause containing at most $k$ as arithmetic variable. A clause is called {\em arithmetically ground} (shorthand:
$a$-ground) if it does not contain arithmetic variables. An arithmetically ground clause is in {\em normal form} if it is irreducible under the defining rewrite
rules. The set of all clauses is denoted by $\clauses$.
\end{definition}

We introduce {\em clause symbols} and denote them by $c,c',c_1,c_2,\ldots$ for defining clause schemata. Clause variables are denoted by $X,Y,X_1,Y_1,\ldots$ and the
set of all clause variables is denoted by $V_c$.

\begin{definition}[clause schema]\label{def.clauseschema}
\begin{itemize}
\item Clauses and clause variables are clause schemata.
\item If $C_1$ and $C_2$ are clause schemata then $C_1 \circ C_2$ is a clause schema.
\item Furthermore, let $c$ be a clause symbol of arity $\beta+\gamma+1$, $a$ an arithmetic term, $x_1,\ldots,x_\beta \in \Vso$ and $X_1,\ldots,X_\gamma \in V_c$.
Then $c(a,x_1,\ldots,x_\beta,X_1,\ldots,X_\gamma)$ is a clause schema w.r.t. the rewrite system $\Rcal(c)$, where $\Rcal(c)$ is of the form
\[
\begin{array}{l}
\{c(\Null,x_1,\ldots,x_\beta,X_1,\ldots,X_\gamma) \to C,\\
c(Sk,x_1,\ldots,x_\beta,X_1,\ldots,X_\gamma) \to c(k,x_1,\ldots,x_\beta,X_1,\ldots,X_\gamma) \circ D\}
\end{array}
\]
where $C$ is an arithmetically ground clause schema s.t. $V(C) \IN \{ x_1, \ldots, x_\beta, \\ X_1, \ldots, X_\gamma \}$ and $D$ is a clause with $V(D) \IN
\{x_1,\ldots,x_\beta,k\}$. The set of clause schemata is denoted by $\CS$.
\end{itemize}
\end{definition}

\begin{example}\label{ex.clschema}
Let $\sigma \in F^3_s, g \in F^1$ ($x \in \Vso,\ l \in V_a$) with the corresponding rewrite rules $\Rcal(\sigma)$
\begin{eqnarray*}
\{\sigma(\Null,x,l) &\to& x(l),\\
\sigma(Sk,x,l) &\to& g(\sigma(k,x,l))\}
\end{eqnarray*}
and let $c(n,x,X)$ be a clause schema for $\Rcal(c)$ consisting of the rules
\begin{eqnarray*}
\{c(0,x,X) &\to& X \circ (\seq P(\sigma(\Null,x,\Null))),\\
c(Sk,x,X) &\to& c(k,x,X) \circ (\seq P(\sigma(Sk,x,Sk)))\}
\end{eqnarray*}
For $X \in V_c$. The normal forms of $c(n,x,X)$ for $\{n \ass \alpha\}$ are just the clause schemata
$$X \seq P(x(0)), P(g(x(1))), \ldots, P(g^\alpha(x(\alpha))).$$
\end{example}

\begin{definition}[semantics of clause schemata]\label{def.clschemasem}
Let $C$ be a clause schema. Let $\vartheta$ be an arithmetically ground $s$-substitution with $V_a(C) \union V_1(C) \IN \dom(\vartheta)$ and $\lambda$ be a
$c$-substitution without clause variables in the range and $V_c(C) \IN \dom(\lambda)$. We define the interpretation of $C$ under $(\vartheta,\lambda)$ as
$$\vc(\vartheta,\lambda,C) = ((C\lambda)\vartheta)\downarrow.$$
\end{definition}

\begin{example}\label{ex.clschemasem}
let $c(n,x,X)$ be the clause schema from Example~\ref{ex.clschema}, $\vartheta = \{n \ass \alpha\}$ and $\lambda = \{X \ass Q(x(n))\}$. Then
$$\vc(\vartheta,\lambda,C) = Q(x(\alpha)) \seq P(x(0)), P(g(x(1))), \ldots, P(g^\alpha(x(\alpha))).$$
\end{example}

Clause schemata define infinite sequences of clauses. For the purpose of schematic CERES it is also vital to describe the infinite sequence of clause sets
$\CL(\varphi_n)$. To this aim we define a formalism for describing sequences of clause sets. Like in the ordinary CERES-method we define a type of clause term, but we
admit variables over finite sets of clauses (clause-set variables), which we denote by $\xi,\xi_0,\xi_1,\ldots$. The set of all clause-set variables is denoted by
$\Vcs$.

\begin{definition}[clause-set term]\label{def.clst}
We define the set of {\em clause-set terms} $\CST$ as follows:
\begin{itemize}
\item if $\xi \in \Vcs$ then $\xi \in \CST$,
\item if $C \in \CS$ then $[C] \in \CST$,
\item if $t_1,t_2 \in \CST$ then $t_1 \Plus t_2 \in \CST$ and $t_1 \Times t_2 \in \CST$.
\end{itemize}
\end{definition}

\begin{definition}\label{def.cltermeval}
Let $t$ be a clause-set term s.t. $V_a(t) \union V_c(t) \union \Vcs(t) = \emptyset$. Then we define the evaluation of $t$ to a set of clauses in the standard way:
\begin{itemize}
\item If $t = [C]$ then $|[C]| = \{C\}$.
\item If $t = t_1 \Plus t_2$ then $|t| = |t_1| \union |t_2|$.
\item If $t = t_1 \Times t_2$ then $|t| = |t_1| \times |t_2|$.
\end{itemize}
\end{definition}

\begin{definition}[semantics of clause-set terms]\label{def.clstsem}
Let $t$ be a clause-set term with $V_c(t) = \{X_1,\ldots,X_\alpha\}$, $\Vcs(t) = \{\xi_1,\ldots,\xi_\beta\}$ and $V_a(t) = \{n\}$. Let $\vartheta = \{n \ass \gamma\}$,
$\lambda = \{X_1 \ass C_1,\ldots,X_\alpha \ass C_\alpha\}$ (for clauses $C_1,\ldots,C_\alpha$) and $\mu = \{\xi_1 \ass s_1,\ldots,\xi_\beta \ass s_\beta\}$
(for clause-set terms $s_1,\ldots,s_\beta$ not containing clause-set variables. Then we define a semantic function
$\vcst$ by
$$\vcst(\vartheta,\lambda,\mu,t) = |(t\mu\lambda\vartheta)\downarrow|.$$
where $|\mbox{ }|$ is from Definition~\ref{def.cltermeval}.
\end{definition}

\begin{example}\label{ex.clst}
Let $c$ be the clause symbol from Example~\ref{ex.clschema}. Then
$$t\colon\ ([c(n,x,X)] \Times [\seq P(x(n))]) \Plus \xi$$
is a clause-set term. Let $\vartheta = \{n \ass \alpha\}$, $\lambda = \{X \ass {\seq}\}$ and $\mu = \{\xi \ass [P(\sigma(n,x,n)) \seq]\}$. Then the evaluation
$\vcst(\vartheta,\lambda,\mu,t)$ is
$$ \{P(g^\alpha(x(\alpha))) \seq \} \union \{\seq P(x(0)), P(g(x(1))), \ldots, P(g^\alpha(x(\alpha))),P(x(\alpha))\}.$$
\end{example}

The definition below is needed to define clause set schemata via clause-set terms.

\begin{definition}\label{def.clstover}
Let $t$ be a clause-set term, $\xi_1,\ldots,\xi_\alpha$ in $\Vcs$, and $s_1,\ldots,s_\alpha$ objects of appropriate type. Then $t\{\xi_1 \ass s_1,\ldots,\xi_\alpha
\ass s_\alpha\}$ is called a clause-set term over $\{s_1,\ldots,s_\alpha\}$ (note that every ordinary clause set term is also a clause set term over any set
$\{s_1,\ldots,s_\alpha\}$).
\end{definition}

\begin{example}\label{ex.clst2}
Let $t$ be the clause-set term
$$t\colon\ ([c(n,x,X)] \Times [\seq P(x(n))]) \Plus \xi$$
from Example~\ref{ex.clst} and $s$ be some object of the type of $\CST$. Then
$$t'\colon\ ([c(n,x,X)] \Times [\seq P(x(n))]) \Plus s$$
is a clause-set term over $\{s\}$.
\end{example}

\begin{definition}[clause-set schema]\label{def.css}
We reserve $d_0,d_1,\ldots$ for denoting clause-set schemata. A {\em clause-set schema} is a tuple $\Delta\colon (d_1,\ldots,d_\alpha)$ together with sets of rewrite
rules $\Rcal(d_1),\ldots,\Rcal(d_\alpha)$ s.t. for all $i=1,\ldots,\alpha$
\[
\begin{array}{l}
\Rcal(d_i) = \{d_i(\Null,x_1,\ldots,x_\delta,X_1,\ldots,X_\beta,\xi_1,\ldots,\xi_\gamma) \to t^b_i,\\
\mbox{   } d_i(Sk,x_1,\ldots,x_\delta,X_1,\ldots,X_\beta,\xi_1,\ldots,\xi_\gamma) \to t^s_i\}
\end{array}
\]
where $t^b_i, t^s_i$ are clause-set terms over terms in $d_1,\ldots,d_\alpha$, $t^b_\alpha$ is a clause-set term and
\begin{eqnarray*}
V(t^b_i) & \IN & \{x_1,\ldots,x_\delta,X_1,\ldots,X_\beta,\xi_1,\ldots,\xi_\gamma\}, \\
V(t^s_i) & \IN & \{x_1,\ldots,x_\delta,X_1,\ldots,X_\beta,\xi_1,\ldots,\xi_\gamma,k\}.
\end{eqnarray*}
Furthermore, we assume that $d_i(\alpha,x_1,\ldots,x_\delta,X_1,\ldots,X_\beta,\xi_1,\ldots,\xi_\gamma)$
is strongly normalizing for all $\alpha\in\nats$.
\end{definition}
Note that the previous definition is more liberal than the definitions of proof schemata
and the schematic language: there, the rewrite rules representing the definitions of the
symbols are required to be primitive recursive, and are therefore strongly normalizing. Here,
we allow any ,,well-formed'', i.e.~strongly normalizing, definition. We will make use of
this more liberal definition in Definition~\ref{def:char_term}, where we will define
a class of clause-set schemata in a mutually recursive way.
\begin{definition}[semantics of clause-set schemata]\label{def.css-semantics}
We extend $\vcst$ to a function $\vcst^*$. Let $\Delta\colon (d_1,\ldots,d_\alpha)$ a clause set schema as in Definition~\ref{def.clstsem}, $\vartheta$ a substitution
on $V_a(\Delta)$, $\lambda$ a substitution on $V_c(\Delta)$ and $\xi$ be a substitution on $\Vcs(\Delta)$. We define
\[
\begin{array}{l}
\vcst^*(\vartheta,\lambda,\xi,d_\alpha(\Null,x_1,\ldots,x_\delta,X_1,\ldots,X_\beta,\xi_1,\ldots,\xi_\gamma)) = \vcst(\vartheta,\lambda,\xi,t^b_\alpha),\\
\vcst^*(\vartheta,\lambda,\xi,d_\alpha(Sk,x_1,\ldots,x_\delta,X_1,\ldots,X_\beta,\xi_1,\ldots,\xi_\gamma)) = \vcst^*(\vartheta,\lambda,\xi,t^s_\alpha).
\end{array}
\]
Note that $t^b_\alpha$ is a clause-set term. For $1 \leq i < \alpha$ we define
\[
\begin{array}{l}
\vcst^*(\vartheta,\lambda,\xi,d_i(\Null,x_1,\ldots,x_\delta,X_1,\ldots,X_\beta,\xi_1,\ldots,\xi_\gamma)) = \vcst^*(\vartheta,\lambda,\xi,t^b_i),\\
\vcst^*(\vartheta,\lambda,\xi,d_i(Sk,x_1,\ldots,x_\delta,X_1,\ldots,X_\beta,\xi_1,\ldots,\xi_\gamma)) = \vcst^*(\vartheta,\lambda,\xi,t^s_i).
\end{array}
\]
The clause set schema defined by $\Delta$ w.r.t. $(\vartheta,\lambda,\mu)$ for $\dom(\vartheta) = \{n\}$ is then defined as
$$\vcst^*(\vartheta,\lambda,\xi,d_1(n,x_1,\ldots,x_\delta,X_1,\ldots,X_\beta,\xi_1,\ldots,\xi_\gamma)).$$
A clause-set schema is called {\em unsatisfiable} if there exist $\lambda$ and $\xi$ s.t. for all $\alpha$ and $\vartheta_\alpha\colon \{n \ass \alpha\}$ the clause
set
$$\vcst^*(\vartheta_\alpha,\lambda,\xi,d_1(n,x_1,\ldots,x_\delta,X_1,\ldots,X_\beta,\xi_1,\ldots,\xi_\gamma))$$
is unsatisfiable.
\end{definition}

\begin{example}\label{ex.css}
Let $\sigma$ be defined by
\begin{eqnarray*}
\{\sigma(\Null,x,l) &\to& x(l),\\
\sigma(Sk,x,l) &\to& g(\sigma(k,x,l))\}
\end{eqnarray*}
where $c(n,x,X)$ is the clause schema from Example~\ref{ex.clschema}, and $\sigma' \in F^1_s$ with the rewrite rules
\begin{eqnarray*}
\{\sigma'(\Null) &\to& a,\\
\sigma'(Sk) &\to& g(\sigma'(k))\}
\end{eqnarray*}
Note that $\sigma'(n)\downarrow_\alpha$ evaluates to $g^\alpha(a)$. Furthermore we define the clause set schema $\Delta =(d_1,d_2)$ by
\begin{eqnarray*}
\Rcal(d_1) &=& \{d_1(\Null,x,X) \to (d_2(\Null,x,X) \Plus \xi),\\
  & & \; d_1(Sk,x,X) \to d_2(Sk,x,X) \Plus [c(Sk,x,X)]\},\\
\Rcal(d_2) &=& \{d_2(\Null,x,X) \to [P(a) \seq],\ d_2(Sk,x,X) \to (d_2(k,x,X) \Plus [P(\sigma'(Sk)) \seq]\}
\end{eqnarray*}
Let $\vartheta = \{n \ass \alpha\}$, $\lambda = \{X \ass \seq\}$ and $\mu = \{\xi \ass [c(\Null,x,X)]\}$; then
$$\vcst^*(\vartheta,\lambda,\mu,d_1(n,x,X)) = \{\seq P(x(0)),\ldots,P(g^\alpha(x(\alpha)));\ P(a) \seq; \ldots, P(g^\alpha(a))\seq \}$$
\end{example}

\begin{definition}[resolution term]
\begin{itemize}
\item clause schemata are resolution terms.
\item Let $s_1$ and $s_2$ be resolution terms w.r.t. $\Rcal_1$ and $\Rcal_2$, and $P$ be an indexed atom. Then $r(s_1;s_2;P)$
is a resolution term w.r.t. $\Rcal_1 \union \Rcal_2$
\end{itemize}
\end{definition}

Resolution terms define resolution deductions only if appropriate substitutions are applied to the clauses unifying atoms in clauses.

\begin{definition}[$\Vso$-substitution schema]
Let $x_1,\ldots,x_\alpha \in \Vso$ and $t_1,\ldots,t_\alpha$ be $s$-terms s.t. $V_a(t_i) \IN \{n,k\}$ for $i=1,\ldots,\alpha$ then
$$\theta\colon \{x_1 \ass \lambda k.t_1,\ldots,x_\alpha \ass \lambda k.t_\alpha\}$$
is called a $\Vso$-substitution schema (note that the terms $t_i$ may contain arbitrary variables in $\Vso$).
\end{definition}

Every $\Vso$-substitution schema evaluates to sequences of ``ordinary'' second order substitutions under an assignment for the parameter $n$. Indeed, let $\vartheta = \{n \ass \beta\}$; then
$$\theta_\beta = \theta\vartheta = \{x_1 \ass \lambda k.(t_1)\downarrow_\beta,\ldots,x_\alpha \ass \lambda k.(t_\alpha)\downarrow_\beta\}$$
Note that the $(t_i)\downarrow_\beta$ contain only $k$ as arithmetic variable.

\begin{example}\label{ex.resterm}
Let $c(n)$ be the clause schema defined in Example~\ref{ex.clschema}. Then the term $t$ defined as
\[
\begin{array}{l}
r(r(c(n,x,X);P(x(n+1)) \seq; P(x(n+1)));P(x(n+2)) \seq; P(x(n+2)))\})
\end{array}
 \]
is a resolution term. We define a $\Vso$-substitution schema $\theta$, s.t. the normal form of $t\{X \ass Q(a) \seq\}\theta\{n \ass \alpha\}$ is a resolution deduction for all $\alpha$.

Let $g \in F^s_2$ s.t. $g$ specifies the primitive recursive function $\gamma$:
\begin{eqnarray*}
\gamma(k,n) &=& 0  \mbox{ for } k<n+1,\\
             &=& 1 \mbox{ for } k =n+1,\\
             &=& 2 \mbox{ for } k > n+1.
\end{eqnarray*}
Such a $g$ exists as all primitive recursive functions can be expressed as schematic terms. $g$ will need also other symbols in $F^s$ for its definition. Let $h,h' \in F^s_4$ be defined as follows:
\[
\begin{array}{l}
h(\bar{0},x,k,n) \to x(k),\ h(l+1,x,k,n) \to h'(l,x,k,n),\\
h'(\bar{0},x,k,n) \to x(\bar{0}),\ h'(l+1,x,k,n) \to \sigma(\bar{1},x,\bar{1}).
\end{array}
\]
We define $\theta = \{x \ass \lambda k.h(g(k,n),x,k,n)\}$.

Then, for all $\vartheta_\alpha\colon \{n \ass \alpha\}$,  $t'_\alpha\colon (t\{X \ass Q(a) \seq\}\theta\vartheta_\alpha)\downarrow$ is indeed a resolution deduction. For $\alpha = 2$ we obtain the resolution term
\[
\begin{array}{l}
r(r(Q(a) \seq P(x(0)),P(g(x(1))),P(g(g(x(2)))); \ P(x(0))\seq; \ P(x(0))); \\
\quad P(g(x(1))) \seq; \ P(g(x(1))))
\end{array}
\]
which represents a resolution deduction of the clause $Q(a) \seq P(g(g(x(2))))$.
\end{example}

\begin{definition}[resolvent]
Let $C\colon C_1 \seq C_2,\ D\colon D_1 \seq D_2$ be clauses with $V_a(\{C,D\}) = \emptyset$ and $V_c(\{C,D\}) = \emptyset$; let $P$ be an atom. Then
$$\res(C,D,P) = C_1,D_1 \setminus P \seq C_2 \setminus P, D_2,$$
where $C\setminus P$ denotes the multi-set of atoms in $C$ after removal of all occurrences of $P$. The clause $\res(C,D,P)$ is called a {\em resolvent} of $C_1$ and $C_2$ on $P$. In case $P$ does not occur in $C_2\theta$ and $D_1\theta$ then $\res(C,D,P)$ is called a pseudo-resolvent (note that inferring $\res(C,D,P)$ from $C$ and $D$ is sound in any case).
\end{definition}

\begin{definition}[resolution deduction]\label{def.resded}
If $C$ is a clause then $C$ is a resolution deduction and $\ES(C) = C$. If $\gamma_1$ and $\gamma_2$ are resolution deductions and $\ES(\gamma_1) = D_1$,
$\ES(\gamma_2) = D_2$ and $\res(D_1,D_2,P) = D$, where $\res(D_1,D_2,P)$ is a resolvent, then $r(\gamma_1,\gamma_2,P)$ is a resolution deduction and $\ES(r(\gamma_1,\gamma_2,P)) = D$.

Let $t$ be a resolution deduction and $\Ccal$ be the set of all clauses occurring in $t$; then $t$ is called a {\em resolution refutation} of $\Ccal$ if $\ES(t) {=}
\seq$.
\end{definition}

Note that resolution terms, containing only ordinary clauses and atoms, represent resolution deductions if, under evaluation of $r$ by $\res$, we obtain a consistent structure of resolvents.

\begin{definition}[tree transformation]\label{def.treetrans}
Any resolution deduction in Definition~\ref{def.resded} can easily be transformed into a resolution tree by the following transformation $T$:
\begin{itemize}
\item If $\gamma = C$ for a clause $C$ then $T(\gamma) = C$.
\item If $\gamma = r(\gamma_1,\gamma_2,P)$, $\varphi_1 = T(\gamma_1)$, $\varphi_2 = T(\gamma_2)$, $\ES(\varphi_1) = C_1$, $\ES(\varphi_2) = C_2$, and
$\res(C_1,C_2,P,\theta)=C$ then $T(\gamma)=$
\[
\infer{C}
 { \deduce{C_1}{(\varphi_1}
   &
   \deduce{C_2}{(\varphi_2)}
 }
\]
\end{itemize}
\end{definition}

The length of $T(\gamma)$ is polynomial in the length of $\gamma$ as can be proved easily.

\begin{example}

$T(r(r(\seq Q_0(x),P_0(x),P_1(x); P_1(x) \seq; P_1(x)); Q_0(x),Q_1(x) \seq; Q_0(x)))$ is the resolution tree
\[
\infer{ \seq P_0(x),Q_1(x) }
 { \infer{ \seq Q_0(x),P_0(x) }
    { \seq Q_0(x),P_0(x),P_1(x) & P_1(x) \seq
    }
    &
   Q_0(x),Q_1(x) \seq
 }
\]
\end{example}

we define a notion of resolution proof schema in the spirit of {\LK}-proof schemata.

\begin{definition}\label{def.rtover}
Let $t$ be a resolution term, $X_1,\ldots,X_\alpha$ in $V_c$, and $s_1,\ldots,s_\alpha$ objects of appropriate type. Then $t\{X_1 \ass s_1,\ldots,X_\alpha \ass
s_\alpha\}$ is called a resolution term over $\{s_1,\ldots,s_\alpha\}$.
\end{definition}

\begin{definition}[resolution proof schema]\label{def.rproofschema}
A resolution proof schema over the variables $x_1,\ldots,x_\alpha \in \Vso$ and $X_1,\ldots, X_\beta \in \Vcs$ is a structure $((\rho_1,\ldots,\rho_\gamma),\Rcal)$ with $\Rcal\colon \Rcal_1 \union \ldots
\union \Rcal_\gamma$, where the $\Rcal_i$ (for $0 \leq i \leq \gamma$) are defined as follows:
$$\Rcal_i = \{\rho_i(0,x_1,\ldots,x_\alpha,X_1,\ldots,X_\beta) \to t^b_i,\
\rho_i(Sk,x_1,\ldots,x_\alpha,X_1,\ldots,X_\beta) \to t^s_i\},$$
where
\begin{itemize}
\item $t^b_i$ is a resolution term over terms of the form $\rho_j(a_j,s_1,\ldots,s_\alpha,C_1,\ldots,C_\beta)$ for $1 \leq i < j$.
\item $t^s_i$  is a resolution term over terms of the form $\rho_j(a_j,s_1,\ldots,s_\alpha,C_1,\ldots,C_\beta)$ and $\rho_i(k,s'_1,\ldots,s'_\alpha,C'_1,\ldots,C'_\beta)$
    for $1 \leq i < j$.
\end{itemize}
\end{definition}

\begin{definition}[semantics of resolution proof schemata]\label{def.rps-semantics}
A resolution proof schema $R$ is called a {\em resolution deduction schema} from a clause-set schema $\Delta$ if there exist  substitutions $\lambda$ for $V_c$ and
$\mu$ for $\Vcs$ and a $\Vso$-substitution schema $\theta$ s.t. for every $\vartheta_\beta$ of the form $\{n \ass \beta\}$
$(\rho_1(n,\bar{x},X_1,\ldots,X_\alpha)\lambda\theta\vartheta_\beta)\downarrow$ is a resolution deduction $t_\beta$ from
$\vcst^*(\vartheta_\beta,\lambda,\mu,d_1(n,\bar{x},Y_1,\ldots,Y_\gamma,\xi_1,\ldots,\xi_\delta))$. If for all $\beta$ $\ES(t_\beta) =\; \seq$ then we call $R$ a
resolution refutation of $\Delta$.
\end{definition}

\begin{example}\label{ex.refschema}
Let $\Delta$ be the clause-set schema defined in Example~\ref{ex.css} defining the sequence of clauses
$$\Ccal_\alpha = \{\seq P(x_0),\ldots,P(g^\alpha(x_\alpha));\ P(a) \seq; \ldots, P(g^\alpha(a))\seq \}$$
Let $(\rho,\Rcal)$ be a proof schema with clause variable $X$ defined by the following rewrite system $\Rcal$:
\[
\begin{array}{l}
\{ \\
\rho(0,x,X) \to r( \seq P(\sigma(\Null,x,\Null)) \circ X; P(\sigma'(\Null)) \seq; P(\sigma(\Null,x,\Null))),\\
\rho(Sk,x,X) \to r(\rho(k,x,\seq P(\sigma(Sk,x,Sk)) \circ X); P(\sigma'(Sk)) \seq; P(\sigma(Sk,x,Sk))).\\
\}
\end{array}
\]
Then $(\rho,\Rcal)$ is a refutation schema for $\Delta$; indeed, for $\lambda = \{X \ass \seq\}$ and $\mu = \{\xi \ass [c(\Null,x,X)]\}$ and $\theta$ defined as
$$\theta = \{x \ass \lambda k.a\}$$
%
we get for all $\vartheta_\alpha = \{n \ass \alpha\}$ that $(\rho(n,x,X)\lambda\theta\vartheta_\alpha)\downarrow$ is a resolution resolution refutation of
$\vcst^*(\vartheta_\alpha,\lambda,\mu,d_1(n,x,X,\xi))$, which is just $\Ccal_\alpha$.
\end{example}

\begin{theorem}\label{the.res-schema-sound}
Resolution refutation schemata are sound, i.e. if $R$ is a resolution refutation schema of a clause-set schema $\Delta$ then $\Delta$ is unsatisfiable.
\end{theorem}
\begin{proof}
Immediate by Definition~\ref{def.rps-semantics} and by Definition~\ref{def.css-semantics}.
\end{proof}

Let us remark here that unsatisfiability of schemata is a property which is not semi-decidable even for propositional schemata (see~\cite{Aravantinos2011}).

\section{The CERES Method for First-Order Schemata}\label{sec:schemceres}

In this section we will consider the problem of cut-elimination for proof schemata.
Note that trivially, for every $\gamma\in\nats$ we can obtain a cut-free
proof of $S(\gamma)$ by computing $\varPsi\downarrow_\gamma$, which contains
cuts, and then applying a usual cut-elimination algorithm. What we are
interested in here is rather a {\em schematic} description of all the cut-free
proofs for a parameter $n$. It is not possible to obtain such a description
by naively applying Gentzen-style cut-elimination to the
{\LKS}-proofs in $\varPsi$, since it is not clear how to handle the case
\begin{center}
  \AxiomC{$(\psi_1(a_1))$}
  \dashedLine
  \UnaryInfC{$\Gamma \seq \Delta, C$}

  \AxiomC{$(\psi_2(a_2))$}
  \dashedLine
  \UnaryInfC{$C, \Pi \seq \Lambda$}

  \RightLabel{$cut$}
  \BinaryInfC{$\Gamma, \Pi \seq \Delta, \Lambda$}
  \DisplayProof
\end{center}
as this would require ``moving the cut through a proof link''.
In this paper, we will go a different route: we will define a CERES method,
which will be based on a {\em global} analysis
of the proof schema. It will eventually yield the desired schematic
description of the sequence of cut-free proofs, as expressed by Theorem~\ref{thm:main}.
\subsection{The Characteristic Term}\label{sec:clause_term}
At the heart of the CERES method lies the {\em characteristic clause set},
which describes the cuts in a proof. 
The connection between cut-elimination and the characteristic clause set is
that any resolution refutation of the characteristic clause set 
can be used as a skeleton of a proof containing
only atomic cuts.

The characteristic
clause set can either be defined directly as in~\cite{Baaz2000},
or it can be obtained
via a transformation from a {\em characteristic term}
as in~\cite{Baaz2006a}. We
use the second approach here; the reason for this will be
explained later.

Our main aim is to extend the usual inductive definition
of the characteristic term to the case of proof links. 
This will give rise
to a notion of {\em schematic characteristic term}. 
The usual definition 
of the characteristic term
depends upon the cut-status of the formula
occurrences in a proof (i.e.~whether a given formula occurrence is a
cut-ancestor, or not). But a formula occurrence in a proof schema gives
rise to many formula occurrences in its evaluation, some of which will
be cut-ancestors, and some will not. Therefore we need some machinery to track
the cut-status of formula occurrences through proof links.
Hence we call
a set $\Omega$ of formula occurrences from the end-sequent of 
an {\LKS}-proof $\pi$ a {\em configuration for $\pi$}.

We will represent the characteristic term of a proof
link in our object language: For all proof symbols $\psi$ and
configurations $\Omega$ we assume a unique 
symbol $\cl^{\psi,\Omega}$
called {\em clause-set symbol}. 
The intended semantics of $\cl^{\psi,\Omega}(a)$ is ``the characteristic clause set of $\psi(a)$, with the configuration $\Omega$''.
\begin{df}[Characteristic term]\label{def:char_term}
Let $\pi$ be an $\LKS$-proof and $\Omega$ a configuration. In the following, by $\Gamma_\Omega, \Delta_\Omega$ and $\Gamma_C, \Delta_C$
we will denote multisets of formulas of $\Omega$- and cut-ancestors respectively.
Let $\rho$ be an inference in $\pi$.
We define the clause-set term $\Theta_\rho(\pi, \Omega)$ inductively:
  \begin{itemize}
    \item if $\rho$ is an axiom of the form $\Gamma_\Omega, \Gamma_C, \Gamma \seq \Delta_\Omega, \Delta_C, \Delta$, then $\Theta_\rho(\pi, \Omega) = [\Gamma_\Omega, \Gamma_C \seq \Delta_\Omega, \Delta_C]$
    \item if $\rho$ is a proof link of the form
 \AxiomC{$(\psi(a, x_1, \ldots , x_\alpha))$}
 \dashedLine
 \UnaryInfC{$\Gamma_\Omega,\Gamma_C,\Gamma\seq\Delta_\Omega,\Delta_C,\Delta$}
 \DisplayProof
then define $\Omega'$ as the set of formula occurrences from $\Gamma_\Omega, \Gamma_C \seq \Delta_\Omega, \Delta_C$ and $\Theta_\rho(\pi,\Omega)= \cl^{\psi,\Omega'}(a,x_1, \ldots , x_\alpha)$
    \item if $\rho$ is a unary rule with immediate predecessor $\rho'$, then $\Theta_\rho(\pi,\Omega) = \Theta_{\rho'}(\pi,\Omega).$
    \item if $\rho$ is a binary rule with immediate predecessors $\rho_1, \rho_2$, then
      \begin{itemize}
       \item if the auxiliary formulas of $\rho$ are $\Omega$- or cut-ancestors, then $\Theta_\rho(\pi,\Omega) = \Theta_{\rho_1}(\pi,\Omega) \oplus \Theta_{\rho_2}(\pi,\Omega)$,
       \item otherwise $\Theta_\rho(\pi,\Omega) = \Theta_{\rho_1}(\pi,\Omega) \otimes \Theta_{\rho_2}(\pi,\Omega).$
      \end{itemize}
  \end{itemize}
Finally, define $\Theta(\pi, \Omega) = \Theta_{\rho_0}(\pi,\Omega)$, where $\rho_0$ is the last inference of $\pi$,
and $\Theta(\pi)=\Theta(\pi,\emptyset)$. $\Theta(\pi)$ is called the characteristic term of $\pi$.
\end{df}

\begin{example}~\label{running_example_ct}
Let us consider the proof schema $\Psi'$ of the sequent $(\forall x) (P(x) \impl P(f(x))) \seq (P(\hat{f}(n,c)) \impl P(g(n,c))) \impl (P(c) \impl P(g(n,c)))$, which is a regularized version of the proof schema $\Psi$ defined in Example~\ref{running_example} (the first-order variables $x_0, x_{k+1}$ are replaced by $x(0), x(k+1)$ respectively). We have two relevant configurations: $\emptyset$ for $\varphi$ and $\Omega = \{ \seq (\forall x) (P(x) \impl P(\hat{f}(n,x)))\}$ for $\psi$. The characteristic terms of $\Psi$ for these configurations are:
\begin{center}
$\begin{array}{lcl}
\Theta(\pi_1, \emptyset) & = & [P(\hat{f}(0,x(0))) \seq P(\hat{f}(0,x(0)))] \oplus ([\seq P(c)] \oplus ([P(\hat{f}(0,c)) \seq] \otimes [\seq] )) \\
\Theta(\nu_1(k), \emptyset) & = & \cl^{\psi, \Omega}(k+1) \oplus ([\seq P(c)] \oplus ([P(\hat{f}(k+1,c)) \seq] \otimes [\seq] ))\\
\Theta(\pi_2, \Omega) & = & [P(\hat{f}(0,x(0))) \seq P(\hat{f}(0,x(0)))] \\
\Theta(\nu_2(k), \Omega) & = & \cl^{\psi, \Omega}(k) \oplus ([P(x(k+1)) \seq P(x(k+1))] \oplus \\ 
  & & \hskip5em ( [P(\hat{f}(k,x(k+1))) \seq] \otimes [\seq P(\hat{f}(k+1,x(k+1)))] ))
 \end{array}$
\end{center}
\end{example}

We say that a clause-set term is {\em normal} if it does not contain 
clause-set symbols and defined function and predicate symbols.
Now we define a notion of characteristic term schema:
\begin{df}[Characteristic term schema]~\label{def:ct_eval}
%
We define the rewrite rules for clause-set symbols for all proof symbols $\psi_\beta$
and configurations $\Omega$:
\[
\cl^{\psi_\beta,\Omega}(0,x_1,\ldots,x_\alpha) \rightarrow \Theta(\pi_\beta,\Omega),\qquad
\cl^{\psi_\beta,\Omega}(k+1,x_1,\ldots,x_\alpha) \rightarrow \Theta(\nu_\beta(k),\Omega),
\]
for all $1\leq\beta\leq\alpha$. Next, let $\gamma\in\nats$
and let $\cl^{\psi_\beta,\Omega}\!\downarrow_\gamma$ be
a normal form of $\cl^{\psi_\beta,\Omega}(\gamma, x_1,\ldots,x_\alpha)$
under the rewrite system just given extended by rewrite rules for defined function and predicate symbols.
Then define $\Theta(\psi_\beta,\Omega)=\cl^{\psi_\beta,\Omega}$ 
and $\Theta(\Psi, \Omega)=\Theta(\psi_1,\Omega)$
and finally the {\em schematic characteristic term}
$\Theta(\Psi)=\Theta(\Psi,\emptyset)$.
\end{df}

We say that a clause-set symbol $\cl^{\psi,\Omega}$ depends on a clause-set symbol $\cl^{\varphi,\Omega'}$, if a term $\Theta(\psi,\Omega)$ contains $\cl^{\varphi,\Omega'}$. We assume that the dependency relation is transitive and reflexive.

The following proposition shows that the definition of the characteristic term schema satisfies the requirement of Definition~\ref{def.css}.

\begin{proposition}
Let $\Psi$ be a proof schema and $\Theta(\Psi)$ be a characteristic term schema of $\Psi$. Then $\Theta(\Psi)$ is strongly normalizing. 
\end{proposition}

\begin{proof}
It is clear that the rewrite rules for defined function and predicate symbols are strongly normalizing, since they are primitive recursive. Also, the rewrite rules of the clause-set symbols for which the dependency relation is acyclic, are strongly normalizing.

Now assume $\cl^{\psi_i,\Omega}$ depends on $\cl^{\psi_j,\Omega'}$ for some $i \not = j$. This means that there is a proof link in $\nu_i(k)$ to $\psi_j$ explicitly or implicitly (i.e in $\nu_i(k)$ there is a proof link to $\psi_{i_1}$, in $\nu_{i_1}(k)$ there is a proof link to $\psi_{i_2}$ and so on. Finally, in $\nu_{i_l}(k)$ there is a proof link to $\psi_j$). In both cases, $\cl^{\psi_j,\Omega'}$ cannot depend on $\cl^{\psi_i,\Omega}$ by the definition of proof schemata. So assume for some $\psi_\beta \in \Psi$, $\cl^{\psi_\beta,\Omega}$ depends on $\cl^{\psi_\beta,\Omega'}$ and vice versa. Then the rewrite rules of $\cl^{\psi_\beta,\Omega}$ and $\cl^{\psi_\beta,\Omega'}$ are still strongly normalizing, since the parameter is strictly decreasing by the definition of proof schemata. \qed
\end{proof}

\begin{example}~\label{running_example_ct_schema}
 Let's consider the proof schema $\Psi'$ and clause-set terms defined in Example~\ref{running_example_ct}. Then the characteristic term schema of $\Psi'$ is $(\cl^{\varphi, \emptyset}, \cl^{\psi, \Omega})$ with the rewrite system:
\begin{center}
$\begin{array}{lcl}
\cl^{\varphi, \emptyset}(0) & \to & [P(\hat{f}(0,x(0))) \seq P(\hat{f}(0,x(0)))] \oplus ([\seq P(c)] \oplus ([P(\hat{f}(0,c)) \seq] \otimes [\seq] )) \\
\cl^{\varphi, \emptyset}(k+1) & \to & \cl^{\psi, \Omega}(k+1) \oplus ([\seq P(c)] \oplus ([P(\hat{f}(k+1,c)) \seq] \otimes [\seq] ))\\
\cl^{\psi, \Omega}(0) & \to & [P(\hat{f}(0,x(0))) \seq P(\hat{f}(0,x(0)))] \\
\cl^{\psi, \Omega}(k+1) & \to & \cl^{\psi, \Omega}(k) \oplus ([P(x(k+1)) \seq P(x(k+1))] \oplus \\ 
  & & \hskip5em ( [P(\hat{f}(k,x(k+1))) \seq] \otimes [\seq P(\hat{f}(k+1,x(k+1)))] ))
 \end{array}$
\end{center}
\end{example}

Now we can explain why we chose to define the characteristic clause
set via the characteristic term: The clause-set term
is closed under the rewrite rules we have given for the clause-set symbols,
while the notion of clause set is not (a clause will in general
become a formula when subjected to the rewrite rules).
%
%
Now, we prove that the notion of characteristic term
is well-defined.
\begin{proposition}
Let $\gamma\in\nats$ and $\Omega$ be a configuration,
then $\Theta(\psi_\beta,\Omega)\downarrow_\gamma$ is a normal clause-set term
for all $1\leq\beta\leq\alpha$.
Hence $\Theta(\Psi)\downarrow_\gamma$ is a normal clause-set term.
\end{proposition}
\begin{proof}
We proceed analogously to the proof of Proposition~\ref{prop:proof_schemata_sound}. \qed
\end{proof}
Next, we show that evaluation and extraction of characteristic terms commute.
We will later use this property to 
derive results on schematic characteristic clause sets from
standard results on (non-schematic) CERES.
\begin{proposition}\label{prop:main}
Let $\Omega$ be a configuration and $\gamma\in\nats$. Then $\Theta(\Psi\downarrow_\gamma, \Omega) = \Theta(\Psi, \Omega)\downarrow_\gamma$.
\end{proposition}

\begin{proof}
We proceed by induction on $\gamma$. If $\gamma = 0$, then $\Theta(\Psi\downarrow_0, \Omega) = \Theta(\pi_1, \Omega)$ and $\Theta(\Psi, \Omega)\downarrow_0 = \Theta(\pi_1, \Omega)$. 

IH1: assume $\gamma > 0$ and for all $\beta < \gamma$, $\Theta(\Psi\downarrow_\beta, \Omega) = \Theta(\Psi, \Omega)\downarrow_\beta$. We proceed by induction on the number $\alpha$ of proof symbols in $\Psi$.

Let $\alpha = 1$. By the definition of characteristic term, constructions of $\Theta(\Psi\downarrow_\gamma, \Omega)$ and $\Theta(\Psi, \Omega)\downarrow_\gamma$ differ only on proof links, i.e. if $(\psi_1(k,x_1,\ldots,x_l))$ is a proof link in $\nu_1(k)$, then by the definition of evaluation of proof schemata, $\Theta(\psi_1\downarrow_\gamma, \Omega)$ contains the term $\Theta(\psi_1\downarrow_\beta, \Omega')$ and by the definition of evaluation of characteristic term schemata, $\Theta(\Psi, \Omega)\downarrow_\gamma$ contains the term $\Theta(\Psi, \Omega')\downarrow_\beta$. Then by the assumption $\Theta(\psi_1\!\downarrow_\beta, \Omega') = \Theta(\Psi, \Omega')\!\downarrow_\beta$ and we conclude that $\Theta(\psi_1\!\downarrow_\gamma, \Omega) = \Theta(\Psi, \Omega)\!\downarrow_\gamma$.

Now, assume $\alpha > 1$ and the proposition holds for all proof schemata with proof symbols less than $\alpha$ (IH2). Again, for proof links in $\nu_1(k)$ of the form $(\psi_1(k, x_1,\ldots,x_l))$ the argument is the same as in the previous case. Let $(\psi_\iota(a, x_1, \\ \ldots, x_l))$, $1< \iota \leq \alpha$, be a proof link in $\nu_1(k)$. Then, again, by the definition of evaluation of proof schemata, $\Theta(\psi_1\!\downarrow_\gamma, \Omega)$ contains the term $\Theta(\psi_\iota\!\downarrow_\lambda, \Omega')$ and by the definition of evaluation of characteristic term schemata, $\Theta(\Psi, \Omega)\!\downarrow_\gamma$ contains the term $\Theta(\varPhi, \Omega')\!\downarrow_\lambda$, where $\varPhi = \left\langle (\pi_\iota, \nu_\iota(k)), \ldots, (\pi_\alpha, \nu_\alpha(k)) \right\rangle$. Clearly, $\varPhi$ contains less than $\alpha$ proof symbols, then by IH2, $\Theta(\psi_\iota\!\downarrow_\lambda, \Omega') = \Theta(\varPhi, \Omega')\!\downarrow_\lambda$ and we conclude that $\Theta(\psi_1\!\downarrow_\gamma, \Omega) = \Theta(\Psi, \Omega)\!\downarrow_\gamma$. \qed
\end{proof}
From the characteristic term we finally define the notion of
characteristic clause set.
%
%
For an \LKS-proof $\pi$ and configuration $\Omega$, $\CL(\pi, \Omega) = |\Theta(\pi,\Omega)|$.
We define the {\em standard characteristic clause set} $\CL(\pi)= \CL(\pi,\emptyset)$
and the {\em schematic characteristic clause set}
$\CL(\Psi,\Omega)=|\Theta(\Psi,\Omega)|$
and
$\CL(\Psi)=\CL(\Psi,\emptyset)$.
%
\begin{example}~\label{running_example_cs}
Let's consider the characteristic term schema defined in Example~\ref{running_example_ct_schema}. Then the sequence of $\CL(\Psi)\!\downarrow_0, \CL(\Psi)\!\downarrow_1, \CL(\Psi)\!\downarrow_2, \ldots$ is: %
$$\begin{array}{l}
\{ P(x_0) \seq P(x_0) \; ; \; \seq P(c) \; ; \; P(c) \seq \}, \\ [1ex]
\{ P(x_0) \seq P(x_0) \; ; \; P(f(x_1)) \seq P(f(x_1)) \; ; \; P(x_1) \seq P(f(x_1)) \; ; \; \seq P(c) \; ; \; P(f(c)) \seq \}, \\ [1ex]
\{ P(x_0) \seq P(x_0) \; ; \; P(f(x_1)) \seq P(f(x_1)) \; ; \; P(f(f(x_2))) \seq P(f(f(x_2))) \; ; \\
\qquad P(x_1) \seq P(f(x_1)) \; ; \; P(f(x_2)) \seq P(f(f(x_2))) \; ; \; \seq P(c) \; ; \; P(f(f(c))) \seq \}, \ldots 
\end{array}$$
After tautology deletion and subsumption the sequence of $\CL(\Psi)\!\downarrow_\gamma$ for $\gamma > 0$ boils down to $\{ P(x_1) \seq P(f(x_1)); \; \seq P(c); \; P(f^\gamma(c)) \seq \}$
\end{example}
%
%
Now we prove the main result about the characteristic clause set and lift it to the schematic case.

\begin{proposition}\label{prop:cl_unsat}
Let $\pi$ be a normal {\LKS}-proof. Then $\CL(\pi)$ is unsatisfiable.
\end{proposition}
\begin{proof}
By the identification of normal {\LKS}-proofs with {\LK}-proofs,
the result follows from Proposition 3.2 in~\cite{Baaz2000}. \qed
\end{proof}
%
%
\begin{proposition}
$\CL(\Psi)\downarrow_\gamma$ is unsatisfiable for all $\gamma \in \nats$ (i.e. $\CL(\Psi)$ is unsatisfiable).
\end{proposition}

\begin{proof}
By Propositions~\ref{prop:main}~and~\ref{prop:proof_schemata_sound}
$\CL(\Psi)\downarrow_\alpha=\CL(\Psi,\emptyset)\downarrow_\alpha=
\CL(\Psi\downarrow_\alpha,\emptyset)=\CL(\Psi\downarrow_\alpha)$
which is unsatisfiable by Proposition~\ref{prop:cl_unsat}. \qed
\end{proof}
%
%
\subsection{Projections}\label{sec:projections}

The next step in the schematization of the CERES method consists in the definition
of schematic proof projections. The aim is, in analogy with the preceding section,
to construct a {\em schematic projection term} that can be evaluated to a set
of normal {\LKS}-proofs. As before, we introduce formal symbols representing sets
of proofs, and again the notion of {\LKS}-proof is not closed under the rewrite
rules for these symbols, which is the reason for introducing the notion
of projection term.

For our term notation we assume for every rule $\rho$ of {\LKS}
a corresponding {\em rule symbol} that, by abuse of notation, we also denote by $\rho$.
Given a unary rule $\rho$ and an {\LKS}-proof $\pi$, there are different ways to apply
$\rho$ to the end-sequent of $\pi$: namely, the choice of auxiliary formulas is free.
Formally, the projection terms we construct will
include this information so that evaluation is always well-defined,
but we will surpress it in the notation since the choice of auxiliary
formulas will always be clear from the context.

For every proof symbol $\psi$ and configuration $\Omega$, we assume a unique
proof symbol $\pr^{\psi,\Omega}$.
Now, a {\em projection term} is a term built inductively from sequents and 
terms $\pr^{\psi,\Omega}(a)$, for some arithmetic expression $a$, 
using unary rule symbols, unary symbols $w^{\Gamma\seq\Delta}$ for all sequents $\Gamma\seq\Delta$
and binary symbols $\oplus, \otimes_\sigma$ for
all binary rules $\sigma$. 
The symbols $\pr^{\psi,\Omega}$ are called {\em projection symbols}.
The intended interpretation of $\pr^{\psi,\Omega}(a)$ is
``the set of characteristic projections of $\psi(a)$, with the configuration $\Omega$''.
\begin{df}[Characteristic projection term]
Let $\pi$ be an $\LKS$-proof and $\Omega$ an arbitrary configuration for $\pi$. Let $\Gamma_\Omega, \Delta_\Omega$ and $\Gamma_C, \Delta_C$ be multisets of formulas corresponding to $\Omega$- and cut-ancestors respectively. We define 
a projection term
$\Xi_\rho(\pi, \Omega)$ inductively:
\begin{itemize}
 \item If $\rho$ corresponds to an initial sequent $S$, then we define $\Xi_\rho(\pi, \Omega) = S.$
 \item If $\rho$ is a proof link in $\pi$ of the form:
 \AxiomC{$(\psi(a, x_1,\ldots,x_\alpha))$}
 \dashedLine
 \UnaryInfC{$\Gamma_\Omega, \Gamma_C, \Gamma \seq \Delta_\Omega, \Delta_C, \Delta$}
 \DisplayProof
then, letting $\Omega'$ be the set of formula occurrences from $\Gamma_\Omega, \Gamma_C \seq \Delta_\Omega, \Delta_C$,
define $\Xi_\rho(\pi,\Omega)=\pr^{\psi,\Omega'}(a, x_1,\ldots,x_\alpha)$.
 \item If $\rho$ is a unary inference with immediate predecessor $\rho'$, then:
  \begin{itemize}
   \item if $\rho$ is $\Ecal$ rule or the auxiliary formula(s) of $\rho$ are $\Omega$- or cut-ancestors, then $\Xi_\rho(\pi,\Omega) = \Xi_{\rho'}(\pi,\Omega),$
   \item otherwise $\Xi_\rho(\pi,\Omega) = \rho( \Xi_{\rho'}(\pi,\Omega) ).$
  \end{itemize}
 \item If $\sigma$ is a binary inference with immediate predecessors $\rho_1$ and $\rho_2$, then:
  \begin{itemize}
   \item if the auxiliary formulas of $\sigma$ are $\Omega$- or cut-ancestors, let $\Gamma_i \seq \Delta_i$ be the ancestors of the end-sequent in the conclusion of $\rho_i$, for $i=1,2$, and define:
    $\Xi_\sigma(\pi,\Omega) = w^{\Gamma_2 \vdash \Delta_2}(\Xi_{\rho_1}(\pi,\Omega)) \oplus w^{\Gamma_1 \vdash \Delta_1}(\Xi_{\rho_2}(\pi,\Omega)),$
   \item otherwise
    $\Xi_\sigma(\pi,\Omega) = \Xi_{\rho_1}(\pi,\Omega) \otimes_{\sigma} \Xi_{\rho_2}(\pi,\Omega).$
  \end{itemize}
\end{itemize}
Define $\Xi(\pi, \Omega) = \Xi_{\rho_0}(\pi,\Omega)$, where $\rho_0$ is the last inference of $\pi$.
\end{df}
We say that a projection term is {\em normal} if it does not contain 
projection symbols.
\begin{example}~\label{running_example_pt}
Let's consider the proof schema $\Psi'$ 
and configurations defined in Example~\ref{running_example_ct}. We introduce the following abbreviations:
\begin{eqnarray*}
A & = & (\forall x)(P(x) \impl P(f(x))) \\
B(n) & = & (P(\hat{f}(n,c)) \impl P(g(n,c))) \impl (P(c) \impl P(g(n,c))) \\
B_1(n) & = & P(\hat{f}(n,c)) \impl P(g(n,c)) \\
B_2(n) & = & P(g(n,c)) 
\end{eqnarray*}
Then the projection terms of $\Psi'$ for those configurations are:
\begin{center}
$\begin{array}{lcl}
\Xi(\pi_1, \emptyset) & = & w^{\seq B(0)}(w_l(P(\hat{f}(0,x(0))) \seq P(\hat{f}(0,x(0))))) \oplus \\
& & w^{A \seq}(\impl_r(\impl_r(w^{B_1(0) \seq B_2(0)}(P(c) \seq P(c)) \oplus \\
& & \quad w^{P(c) \seq}(P(\hat{f}(0,c)) \seq P(\hat{f}(0,c)) \otimes_{\impl_l} P(g(0,c))\seq P(g(0,c)) ) ))) \\
\Xi(\nu_1(k), \emptyset) & = & w^{\seq B(k+1)}(\pr^{\psi, \Omega}(k+1)) \oplus \\
& & w^{A \seq}(\impl_r(\impl_r( w^{B_1(k+1) \seq B_2(k+1)}(P(c) \seq P(c)) \oplus \\
& & \quad w^{P(c) \seq}(P(\hat{f}(k+1,c)) \seq P(\hat{f}(k+1,c)) \otimes_{\impl_l} \\
& & \hskip10em P(g(k+1,c))\seq P(g(k+1,c)) ) ))) \\
\Xi(\pi_2, \Omega) & = & w_l(P(\hat{f}(0,x(0))) \seq P(\hat{f}(0,x(0)))) \\
\Xi(\nu_2(k), \Omega) & = & c_l(w^{A \seq}(\pr^{\psi, \Omega}(k)) \oplus w^{A \seq}( \\
& & \qquad w^{A \seq}(P(x(k+1)) \seq P(x(k+1))) \oplus \\
& & \qquad w^{\seq}(\forall_l(P(\hat{f}(k,x(k+1))) \seq P(\hat{f}(k,x(k+1))) \otimes_{\impl_l} \\
& & \hskip7em P(\hat{f}(k+1,x(k+1))) \seq P(\hat{f}(k+1,x(k+1)))) )) )
 \end{array}$
\end{center}
\end{example}
We now define the projection-set schema, which is compatible with the respective definition
for clause-set terms.
\begin{df}[Projection-set schema]
We define the rewrite rules for projection term symbols for all proof symbols $\psi_\beta$
and configurations $\Omega$:
\[
\pr^{\psi_\beta,\Omega}(0, x_1,\ldots,x_\alpha) \rightarrow \Xi(\pi_\beta,\Omega),\qquad
\pr^{\psi_\beta,\Omega}(k+1, x_1,\ldots,x_\alpha) \rightarrow \Xi(\nu_\beta(k),\Omega),
\]
for all $1\leq\beta\leq\alpha$. Next, let $\gamma\in\nats$
and let $\pr^{\psi_\beta,\Omega}\!\downarrow_\gamma$ be
a normal form of $\pr^{\psi_\beta,\Omega}(\gamma, x_1,\ldots,x_\alpha)$
under the rewrite system just given extended by rewrite rules for defined function and predicate symbols.
Then define $\Xi(\psi_\beta,\Omega)=\pr^{\psi_\beta,\Omega}$ 
and $\Xi(\Psi, \Omega)=\Xi(\psi_1,\Omega)$
and finally the {\em schematic projection term}
$\Xi(\Psi)=\Xi(\Psi,\emptyset)$.

%
\end{df}
%
%
\begin{proposition}~\label{prop:proj_main}
Let $\Omega$ be a configuration and $\gamma\in\nats$. Then $\Xi(\Psi\!\downarrow_\gamma, \Omega) = \Xi(\Psi, \Omega)\!\downarrow_\gamma$. 
\end{proposition}

\begin{proof}
We proceed as in the proof of Proposition~\ref{prop:main}. \qed
\end{proof}

We will define a map from normal projection terms to sets of normal {\LKS}-proofs. For this, we need some auxiliary notation.
The discussion
regarding the notation for the application of rules from the beginning of this section
applies here.
\begin{df}
Let $\rho$ be a unary and $\sigma$ a binary rule. Let $\varphi, \pi$ be \LKS-proofs, then $\rho(\varphi)$ is the \LKS-proof obtained from $\varphi$ by applying $\rho$, and $\sigma(\varphi,\pi)$ is the proof obtained from the proofs $\varphi$ and $\pi$ by applying $\sigma$.

Let $P, Q$ be sets of \LKS-proofs. Then $\rho(P) = \{ \rho(\pi) \mid \pi \in P \}$, $P^{\Gamma \seq \Delta} = \{ \pi^{\Gamma \seq \Delta} \mid \pi \in P\}$, where $\pi^{\Gamma \seq \Delta}$ is $\pi$ followed by weakenings adding $\Gamma \seq \Delta$,
and $P \times_\sigma Q = \{ \sigma(\varphi, \pi) \mid \varphi \in P, \pi \in Q \}$. 
\end{df}

\begin{df}~\label{def:proj_eval}
Let $\Xi$ be a normal projection term. Then we define a set of normal {\LKS}-proofs $|\Xi|$ in the following way:
\begin{itemize}
 \item $| A \seq A| = \{ A \seq A \}$,
 \item $|\rho(\Xi)| = \rho(|\Xi|)$ for unary rule symbols $\rho$,
 \item $|w^{\Gamma \seq \Delta}(\Xi)| = |\Xi|^{\Gamma \seq \Delta}$,
 \item $|\Xi_1 \oplus \Xi_2| = |\Xi_1| \cup |\Xi_2|$,
 \item $|\Xi_1 \otimes_\sigma \Xi_2| = |\Xi_1| \times_\sigma |\Xi_2|$ for binary rule symbols $\sigma$.
\end{itemize}
For normal {\LKS}-proofs $\pi$ and configurations $\Omega$ we define $\PR(\pi, \Omega) = |\Xi(\pi,\Omega)|$
and the {\em standard projection set} $\PR(\pi)=\PR(\pi,\emptyset)$. 
For $\gamma\in\nats$ we define $\PR(\Psi)\downarrow_\gamma=|\Xi(\Psi)\downarrow_\gamma|$.
\end{df}
%
%
The following result describes the relation between the standard projection set and
characteristic clause set in the normal case. It will allow us to construct, together with a resolution refutation of $\CL(\Psi)$, essentially cut-free proofs of $S(\gamma)$ for all $\gamma \in \nats$.
Finally, the result is lifted to the schematic case.
\begin{proposition}\label{prop:proj_corr}
Let $\pi$ be a normal {\LKS}-proof with end-sequent $S$, then for all clauses $C \in \CL(\pi)$, there exists a normal
{\LKS}-proof $\pi \in \PR(\pi)$ with end-sequent $S \circ C$.
\end{proposition}

\begin{proof}
By the identification of normal {\LKS}-proofs with {\LK}-proofs,
the result follows from Definition~\ref{def:proj_eval} and Lemma 3.1 in~\cite{Baaz2000}. \qed 
\end{proof}
%
%
\begin{proposition}\label{prop:proj_set_main}
Let $\gamma\in\nats$, then $\PR(\Psi\!\downarrow_\gamma) = \PR(\Psi)\downarrow_\gamma$.
\end{proposition}
\begin{proof}
This result follows directly from Proposition~\ref{prop:proj_main}.
\qed
\end{proof}
\begin{proposition}
Let $\gamma\in\nats$, then for every clause $C\in \CL(\Psi)\!\downarrow_\gamma$
there exists a normal {\LKS}-proof $\pi\in \PR(\Psi)\!\downarrow_\gamma$
with end-sequent $C\circ S(\gamma)$.
\end{proposition}
\begin{proof}
By Proposition~\ref{prop:main}, $\CL(\Psi)\!\downarrow_\gamma = \CL(\Psi\!\downarrow_\gamma)$,
and by Proposition~\ref{prop:proj_set_main}, $\PR(\Psi)\!\downarrow_\gamma$ $= \PR(\Psi\!\downarrow_\gamma)$.
Then the result follows from Proposition~\ref{prop:proj_corr}, since $\Psi\!\downarrow_\gamma$
has end-sequent $S(\gamma)$ by definition.
\qed
\end{proof}

\subsection{ACNF Schema}

To produce an Atomic Cut Normal Form, we need to transform a resolution refutation into an \LKS-proof skeleton. Then the ACNF is produced by substituting each clause at the leaf nodes of this skeleton by the corresponding projections and appending necessary contractions at the end of the proof.

\begin{df}[Transformation]
Let $\varrho$ be a ground resolution refutation. Then the transformation $TR(\varrho)$ is defined inductively:
\begin{itemize}
 \item if $\varrho = C$ for a clause $C$, then $TR(\varrho) = C$,
 \item if $\varrho = r(\varrho_1;\varrho_2;P)$, then $TR(\varrho)$ is:
\begin{prooftree}
 \AxiomC{$(TR(\varrho_1))$}
\noLine
\UnaryInfC{$\Gamma \seq \Delta, P, \ldots,P$}
\doubleLine
\RightLabel{$c \colon r*$}
\UnaryInfC{$\Gamma \seq \Delta,P$}

 \AxiomC{$(TR(\varrho_2))$}
\noLine
\UnaryInfC{$P, \ldots,P, \Pi \seq \Lambda$}
\doubleLine
\RightLabel{$c \colon l*$}
\UnaryInfC{$P, \Pi \seq \Lambda$}

\RightLabel{$cut$}
\BinaryInfC{$\Gamma, \Pi \seq \Delta, \Lambda$}
\end{prooftree}
\end{itemize}
\end{df}

\begin{example}
Let us compute the ACNF of the proof schema $\Psi$ defined in Example~\ref{running_example}. First we should give a resolution refutation schema for the characteristic clause set defined in Example~\ref{running_example_cs}. Let $R = ((\varrho, \delta), \Rcal)$, where $\Rcal$ is the following rewriting system: 

\begin{eqnarray*}
\varrho(0,x) & \ra & r(\delta(0,x); \; P(\hat{f}(0,c)) \seq; \; P(\hat{f}(0,c))), \\
\varrho(k+1,x) & \ra & r(\delta(k+1,x); \; P(\hat{f}(k+1,c)) \seq; \; P(\hat{f}(k+1,c))), \\
\delta(0,x) & \ra & \seq P(c), \\
\delta(k+1,x) & \ra & r(\delta(k,x); \; P(x(k+1)) \seq P(f(x(k+1))); \; P(\hat{f}(k,c)))
\end{eqnarray*}

Next we define a predecessor function. Let $pre \colon \omega \to \omega$ be a defined function symbol, then we define function $pre(n)$ with the rewrite rules $pre(0) \to 0$ and $pre(k+1) \to k$. Now we define the $\Vso$-substitution $\theta = \{x \la \lambda k. \hat{f}(pre(k),c)\}$; then $\varrho(n,x)\theta\!\downarrow_\gamma$ is a resolution refutation schema for all $\gamma \in \nats$. Finally, we compute $TR(\varrho(n,x)\theta\!\downarrow_0)$:
\begin{prooftree}
 \AxiomC{$\seq P(c)$}

 \AxiomC{$P(c) \seq$}
\RightLabel{$cut$}
\BinaryInfC{$\seq$}
\end{prooftree}
and ACNF of $\Psi$ for $n=0$ is (we denote the sequent $S(0) \colon (\forall x) (P(x) \impl P(f(x))) \seq (P(c) \impl P(g(0,c))) \impl (P(c) \impl P(g(0,c)))$ with $A \seq B$):
\begin{prooftree}
\scriptsize
\AxiomC{$P(c) \seq P(c)$}
\doubleLine
\RightLabel{$w \colon l,r$}
\UnaryInfC{$P(c) \impl P(g(0,c)), P(c) \seq P(c), P(g(0,c))$}
\RightLabel{$\impl \colon r$}
\UnaryInfC{$P(c) \impl P(g(0,c)) \seq P(c), P(c) \impl P(g(0,c))$}
\RightLabel{$\impl \colon r$}
\UnaryInfC{$\seq P(c), B$}
\RightLabel{$w \colon l$}
\UnaryInfC{$A \seq P(c), B$}

\AxiomC{$P(c) \seq P(c)$}
\AxiomC{$P(g(0,c)) \seq P(g(0,c))$}
\RightLabel{$\impl \colon l$}
\BinaryInfC{$P(c) \impl P(g(0,c)), P(c) \seq P(g(0,c))$}
\RightLabel{$w \colon l$}
\UnaryInfC{$P(c), P(c) \impl P(g(0,c)), P(c) \seq P(g(0,c))$}
\RightLabel{$\impl \colon r$}
\UnaryInfC{$P(c) \impl P(g(0,c)), P(c) \seq P(c) \impl P(g(0,c))$}
\RightLabel{$\impl \colon r$}
\UnaryInfC{$P(c) \seq B$}
\RightLabel{$w \colon l$}
\UnaryInfC{$P(c), A \seq B$}
 
\RightLabel{$cut$}
\BinaryInfC{$A, A \seq B, B$}
\doubleLine
\RightLabel{$c \colon l,r$}
\UnaryInfC{$(\forall x) (P(x) \impl P(f(x))) \seq (P(c) \impl P(g(0,c))) \impl (P(c) \impl P(g(0,c)))$}
\end{prooftree}
\end{example}

Finally, we can summarize the CERES method of cut-elimination for proof schemata by defining the whole CERES-procedure CERES-s on schemata (where $\Psi$ is a proof schema):
\\[1ex]
Phase 1 of CERES-s: (schematic construction)
\begin{itemize}
\item compute $\CL(\Psi)$;
\item compute $\PR(\Psi)$;
\item construct a resolution refutation schema $R = ((\varrho_1,\ldots,\varrho_\beta),\Rcal)$ of $\CL(\Psi)$, and a $V_c$-substitution $\lambda$ and a $\Vso$-substitution $\theta$ according to Definition~\ref{def.rps-semantics}.
\end{itemize}
Phase 2 of CERES-s: (evaluation, given a number $\alpha$)
\begin{itemize}
\item compute $\CL(\Psi)\!\downarrow_\alpha$;
\item compute $\PR(\Psi)\!\downarrow_\alpha$;
\item compute $\varrho_1(n,\bar{x},\bar{X})\lambda\theta\!\downarrow_\alpha$ and $T_\alpha \colon TR(\varrho_1(n,\bar{x},\bar{X})\lambda\theta\!\downarrow_\alpha)$;
\item append the corresponding projections in $\PR(\Psi)\!\downarrow_\alpha$ to $T_\alpha$ and propagate the contexts down in the proof.
\end{itemize}

\begin{theorem}\label{thm:main}
Let $\Psi$ be a proof schema with end-sequent $S(n)$. 
Then the evaluation of CERES-s
produces for all $\alpha\in\nats$ a ground {\LKS}-proof 
$\pi$ of $S(\alpha)$ with at most atomic cuts such that 
its size $|\pi|$ polynomial in $|\varrho_1(n,\bar{x},\bar{X})\lambda\theta\!\downarrow_\alpha| \cdot |\PR(\Psi)\!\downarrow_\alpha|$.
\end{theorem}
\begin{proof}
Let $\alpha \in \nats$. By Proposition~\ref{prop:proj_set_main} we obtain for any clause in $\Ccal_\alpha\colon \CL(\Psi)\!\downarrow_\alpha$ a corresponding
projection of the ground proof $\psi_\alpha$ in $\PR(\Psi)\!\downarrow_\alpha$. Let $R = ((\varrho_1,\ldots,\varrho_\beta),\Rcal)$ be a resolution refutation
schema for $\CL(\Psi)$ constructed in phase 1 of CERES-s and $T_\alpha$ the corresponding tree. Clearly the length of any projection is at most $|\PR(\Psi)\!\downarrow_\alpha|$ and $|T_\alpha|$ is polynomial in $\varrho_1(n,\bar{x},\bar{X})\lambda\theta\!\downarrow_\alpha$. Moreover, the resulting proof $\pi_\alpha$ of $S(\alpha)$ obtained in the last step of phase 2 contains at most atomic cuts.
\end{proof}

\section{Open Problems}
The current results obtained by $\CERESs$ can be considered as a first step of performing cut-elimination in inductive proofs. Currently our formalism admits just one parameter and thus is not capable of modeling nested inductions. Hence a generalization of the method to several parameters is highly desirable. While the construction of the schematic characteristic clause sets and of the schematic proof projections is fully mechanizable (and already implemented\footnote{\url{http://www.logic.at/ceres/system/gapt.html}}) a fully automated construction of schematic resolution refutations is impossible even in principle. However, for practical proof analysis of nontrivial proofs an interactive use of the schematic resolution calculus and a formal verification of the obtained proofs would be vital. The current schematic resolution method is very strong and the computation of the schematic most general unifiers is undecidable. It would be useful to also search for weaker systems admitting a higher degree of automation which are still capable of formalizing relevant problems. 
Towards an interpretation of the datastructures of CERES in practical applications, further work
needs to be done: In~\cite{HLWWP08c} it was shown that {\em Herbrand sequents} are a useful tool for the interpretation
of formal proofs by humans. In the context of the present work, a suitable notion of schematic Herbrand sequent
should be defined and it should be shown how to extract such a sequent from the data structures used by CERES
(i.e.~from a resolution refutation schema of $\CL(\Psi)$, and from $\PR(\Psi)$).
A concrete application of $\CERESs$ would be the full formalization and verification of the schematic analysis of F\"urstenberg's proof of the infinity of primes shown in~\cite{Baaz2008}. In this paper the CERES method was applied to an infinite sequence of proofs; the sequence of characteristic clause sets was found empirically  and the refutation of the infinite sequence of characteristic clause sets was performed on the mathematical meta-level. In contrast $\CERESs$ is capable of defining the schematic characteristic clause sets $\Ccal_n$ and the projections fully automatic and provides a formalism for refuting $\Ccal_n$ formally. We believe that (the current implementation of) $\CERESs$ can also serve as tool for a semi-automated development of proof schemata by mathematicians.

\bibliographystyle{plain}
\bibliography{literatur}

\end{document}